%% file: mainJSACv12.tex
\documentclass[journal,draftclsnofoot,onecolumn,12pt,twoside]{IEEEtran}
\normalsize

\ifCLASSINFOpdf

\else

   \usepackage[dvips]{graphicx}

\fi

\usepackage{booktabs}
\usepackage{afterpage}
\usepackage{graphicx}
\usepackage[font=small]{caption}
\usepackage{subfig}
\usepackage{t1enc}
\usepackage{tabularx}
\usepackage[table]{xcolor}
\usepackage{amsmath,amsfonts, amssymb, amsthm}
\usepackage{float}
\usepackage{booktabs}
\usepackage{multicol}
\usepackage{multirow}
\usepackage{epstopdf}

\usepackage{url}
\usepackage{todonotes}

\usepackage[nolist]{acronym}
\input{acronym}

\newtheorem{thm}{Theorem}

\newtheorem{lem}[thm]{Lemma}

\DeclareMathSizes{12}{10}{8}{6} 

\newcommand{\FF}{{\mathbb{F}}}
\newcommand{\ev}{{\text{ev}}}
\newcommand{\sub}{{\text{SubfieldSubcode}}_2}
\newcommand{\modulo}{{\text{~mod~}}}

\newcommand{\colvec}[2][.8]{%
  \scalebox{#1}{%
    \renewcommand{\arraystretch}{.8}%
    $\begin{bmatrix}#2\end{bmatrix}$%
  }
}

\begin{document}
\title{
Fulcrum Network Codes: A Code for Fluid Allocation of Complexity \vspace{-0.2cm}
}


\author{Daniel E. Lucani, Morten V. Pedersen, Diego Ruano, Chres
  W. S{\o}rensen, Frank H. P. Fitzek, Janus Heide, Olav
  Geil \vspace{-2.1cm} \thanks{D. E. Lucani, M. V. Pedersen, and
    C. W.S{\o}rensen are with Department of Electronic Systems,
    Aalborg University, Denmark, e-mail: \{del, mvp,
    cws\}@es.aau.dk. D. Ruano and O. Geil are with the Department of
    Mathematical Sciences, Aalborg University, Denmark, e-mail: \{
    diego, olav\}@math.aau.dk. F. H. P. Fitzek is with the Chair of
    Communication Networks, Technische Universitaet Dresden, Germany,
    e-mail: frank.fitzek@tu-dresden.de. J. Heide is with Steinwurf
    ApS, Denmark, e-mail: janus@steinwurf.com. This work was
    partially financed by the Green Mobile Cloud project (Grant
    No. DFF - 0602-01372B), the Colorcast project (Grant No. DFF -
    0602-02661B), the TuneSCode project (Grant No. DFF - 1335-00125),
    and project Grant No. DFF-4002-00367 granted by the Danish Council
    for Independent Research. The work was also partially supported by
    the Spanish MINECO project (Grant No. MTM2012-36917-C03-03).}  }

\maketitle
\begin{abstract}

  This paper proposes Fulcrum network codes, a network coding
  framework that achieves three seemingly conflicting objectives: (i)
  to reduce the coding coefficient overhead to almost n bits per
  packet in a generation of n packets; (ii) to operate the network
  using only $GF(2)$ operations at intermediate nodes if necessary,
  dramatically reducing complexity in the network; (iii) to deliver an
  end-to-end performance that is close to that of a high-field network
  coding system for high-end receivers while simultaneously catering
  to low-end receivers that decode in $GF(2)$. As a consequence of
  (ii) and (iii), Fulcrum codes have a unique trait missing so far in
  the network coding literature: they provide the network with the
  flexibility to spread computational complexity over different
  devices depending on their current load, network conditions, or even
  energy targets in a decentralized way. At the core of our framework
  lies the idea of precoding at the sources using an expansion field
  $GF(2^h)$ to increase the number of dimensions seen by the network
  using a linear mapping. Fulcrum codes can use any high-field linear
  code for precoding, e.g., Reed-Solomon, with the structure of the
  precode determining some of the key features of the resulting
  code. For example, a systematic structure provides the ability to
  manage heterogeneous receivers while using the same data stream. Our
  analysis shows that the number of additional dimensions created during precoding controls the trade-off
  between delay, overhead, and complexity. Our
  implementation and measurements show that Fulcrum achieves similar
  decoding probability as high field \ac{RLNC} approaches but with
  encoders/decoders that are an order of magnitude faster.

\end{abstract}
\IEEEpeerreviewmaketitle

\vspace{-0.4cm}
\section{Introduction}
\vspace{-0.2cm}

Ahlswede et al~\cite{ahlswede} proposed network coding (NC) as a means
to achieve network capacity of multicast sessions as determined by the
min-cut max-flow theorem~\cite{Elias56}, a feat that was provably
unattainable using standard store-and-forwarding of packets
(routing). NC breaks with this paradigm, encouraging intermediate
nodes in the network to mix (recode) data packets. Thus, network
coding proposed a store-code-forward paradigm to network operation,
essentially extending the set of functions assigned to intermediate
nodes to include coding operations. Linear network codes were shown to
be sufficient to achieve multicast capacity~\cite{Li03}. \ac{RLNC}
provides an asymptotically optimal and distributed approach to create
linear combinations using random coefficients at intermediate
nodes~\cite{bib_rand_nc}.

In fact, network coding has shown significant gains in a multitude of
settings, from wireless networks~\cite{Zhao10, catwoman, Nistor11},
and multimedia transmission~\cite{Seferoglu09}, to distributed
storage~\cite{Dimakis11}, and Peer-to-Peer (P2P)
networks~\cite{Gkantsidis06}. Practical implementations have also
confirmed NC's gains and
capabilities~\cite{xoreInAir,Pedersen09,Vingelmann11}. The reason
behind these gains lies in two facts. First, the network need not
transport each packet without modification through the network, which
opens more opportunities and freedom to deliver the data to the
receivers and increases the impact of each transmitted coded packet (a
linear combination of the original packets). Second, receivers no
longer need to track individual packets, but instead accumulate enough
independent linear combinations in order to recover the original
packets (decode). These relaxations have a profound impact on system
designs and achievable gains.

\begin{figure}[t]
\centering
\includegraphics[width=0.5\textwidth]{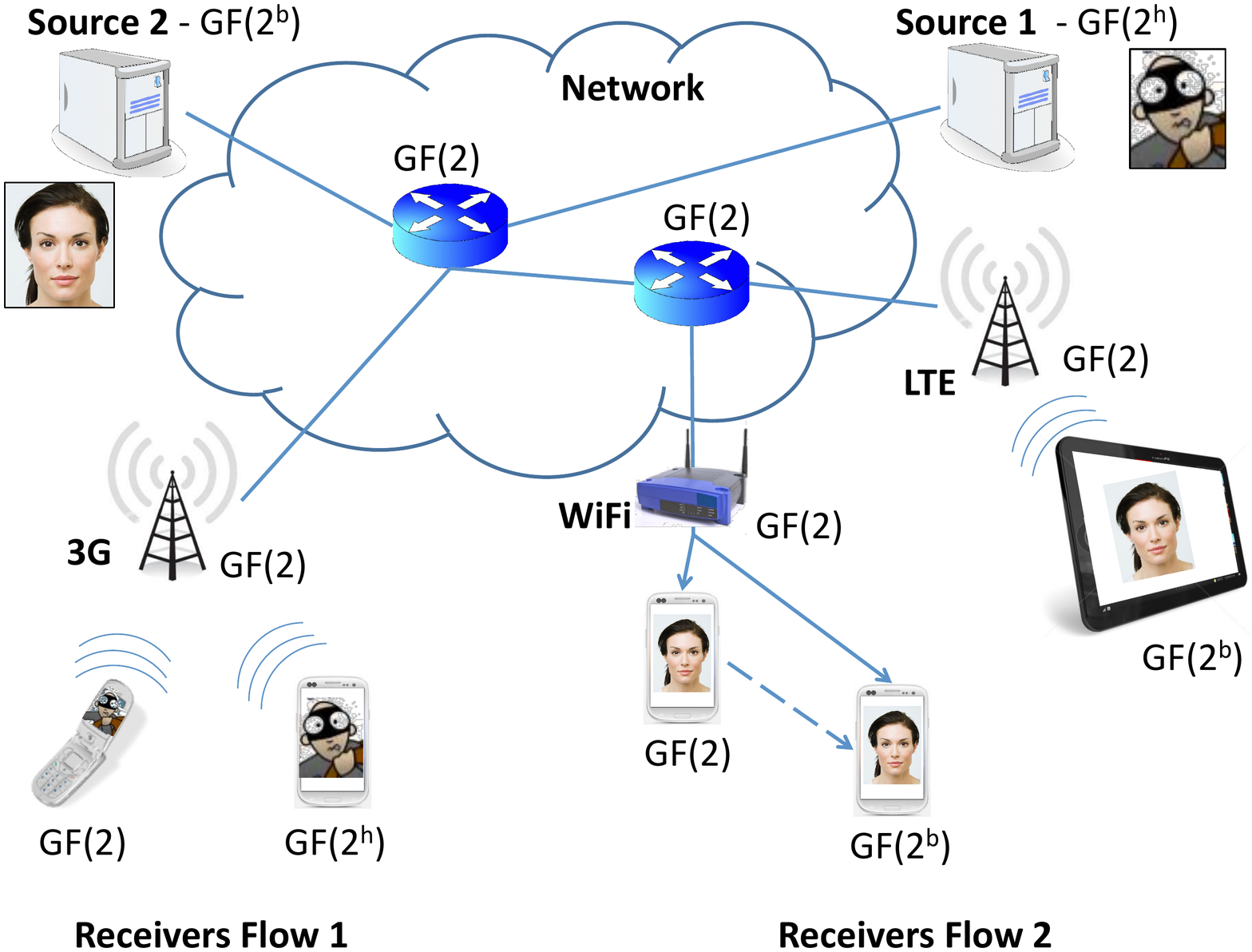}
\caption{ Fulcrum network codes allow sources and receivers to operate
  at higher field sizes to achieve high performance but maintaining
  compatibility with the $GF(2)$--only network. Receivers can choose
  to trade--off delay with decoding effort by choosing to decode with
  $GF(2)$ or in higher fields. }
\label{fig:General}
\vspace{-0.5cm}
\end{figure}

After more than a decade of research and in spite of NC's theoretical
gains in throughput, delay, and energy performance, its wide spread
assimilation remains elusive. One, if not the most, critical weakness
of the technology is the inherent complexity that it introduces into
network devices. This complexity is driven by two factors. First,
devices must perform additional processing, which may limit the energy
efficiency gains or even become a bottleneck in the system's overall
throughput if processing is slower than the incoming/outgoing data
rates. This additional effort can be particularly onerous if we
consider that the conventional wisdom dictates that large field sizes
\textit{are needed} to provide high reliability, throughput, and delay
performance. In addition to the computational burden, the use of high
field sizes comes at the cost of a higher signaling overhead to
communicate the coefficients used for coding the data packets. Other
alternatives, e.g., sending a seed for a pseudo-random number
generator, are relevant end-to-end but do not allow for a simple
recoding mechanism.  Interestingly,~\cite{Heide11} showed that using
moderate field sizes, specially $GF(2)$~\footnote{We shall use $GF(q)$
  and $\FF_q$ to identify finite fields of size $q$.}, is key to
achieving a reasonable trade-off among computational complexity,
throughput performance, and total overhead specially when recoding
data packets. This is particularly encouraging since $GF(2)$
performing encoding/decoding could be as fast as $160$~Mbps and
$9600$~Mbps in a 2009 mobile phone and laptop~\cite{Heide09}, respectively, while in
2013 the speeds increased by five-fold in high-end
phones~\cite{LeanMean}. Even limited sensors, e.g., TelosB motes, can
generate packets in $GF(2)$ at up to $500$~kbps~\cite{Nistor12}.

Second, devices must support different configurations for each
application or data flow, e.g., different field sizes, to achieve a
target performance. Supporting disparate configurations translates
into high costs in hardware, firmware, or software. In computationally
constrained devices, e.g., sensors, the support for encoding,
recoding, or decoding in higher fields is prohibitive due to the
processing effort required. On the other end of the spectrum,
computationally powerful devices may also be unable to support
multiple configurations. For example, high-load, high-speed Internet
routers would require deep packet inspection to determine the coding
configuration, followed by a different treatment of each incoming
packet. This translates into additional expensive hardware to provide
high-processing speeds.  Additionally, intermediate nodes in the
network are heterogeneous in nature, which limits the system's viable
configurations.

A separate, yet related practical issue is the fact that receivers
interested in the same data flow may have wildly different
computational, display, and battery capabilities as well as different
network conditions. This end-device heterogeneity may restrict service
quality at high-end devices when support is required for low-end
devices, may deny service to low end devices for the benefit of high
end ones, or require the system to invest additional resources
supporting parallel data flows, each with characteristics matching
different sets of users.

A clear option to solve the compatibility and complexity challenges is
to limit sources, intermediate nodes, and receivers to use only
$GF(2)$. However, using only $GF(2)$ may deprive higher-end devices of
achieving higher reliability and throughput performance. Is it
possible to provide a single, easily implementable, and compatible
network infrastructure that supports flows with different end-to-end
requirements?

This paper shows that the solution, called Fulcrum network codes, is
simple, tunable, and surprisingly powerful. This is a framework that
hinges on using only $GF(2)$ operations in the network
(Fig.~\ref{fig:General}), to achieve reduced overhead, computational
cost, and compatibility to heterogeneous devices and data flows in the
network, while providing the opportunity of employing higher fields
end-to-end via a tunable and straightforward precoding mechanism for
higher performance. Fig.~\ref{fig:General} shows an example, where two
sources operate using different fields $GF(2^h)$ and $GF(2^b)$ for
source 1 and 2, respectively. The intermediate nodes in the network
use only $GF(2)$ operations. With Fulcrum network codes, the left-most
receiver of flow 1 in Fig.~\ref{fig:General} can choose to decode
using $GF(2)$ only as it has limited computation capabilities. Since
the left-most receiver of flow 2 has a better channel than other
devices and the router may have to broadcast for a longer time due to
the other receivers, the left-most receiver can choose to save energy
on computation by accumulating additional packets and decoding using
$GF(2)$. Furthermore, this receiver can also recode packets and
send them to a neighbor interested in the same content, thus
increasing the coverage of the system and reducing the number of
transmissions needed to deliver it.

\begin{figure*}[t]
\centering
\includegraphics[width=0.7\textwidth]{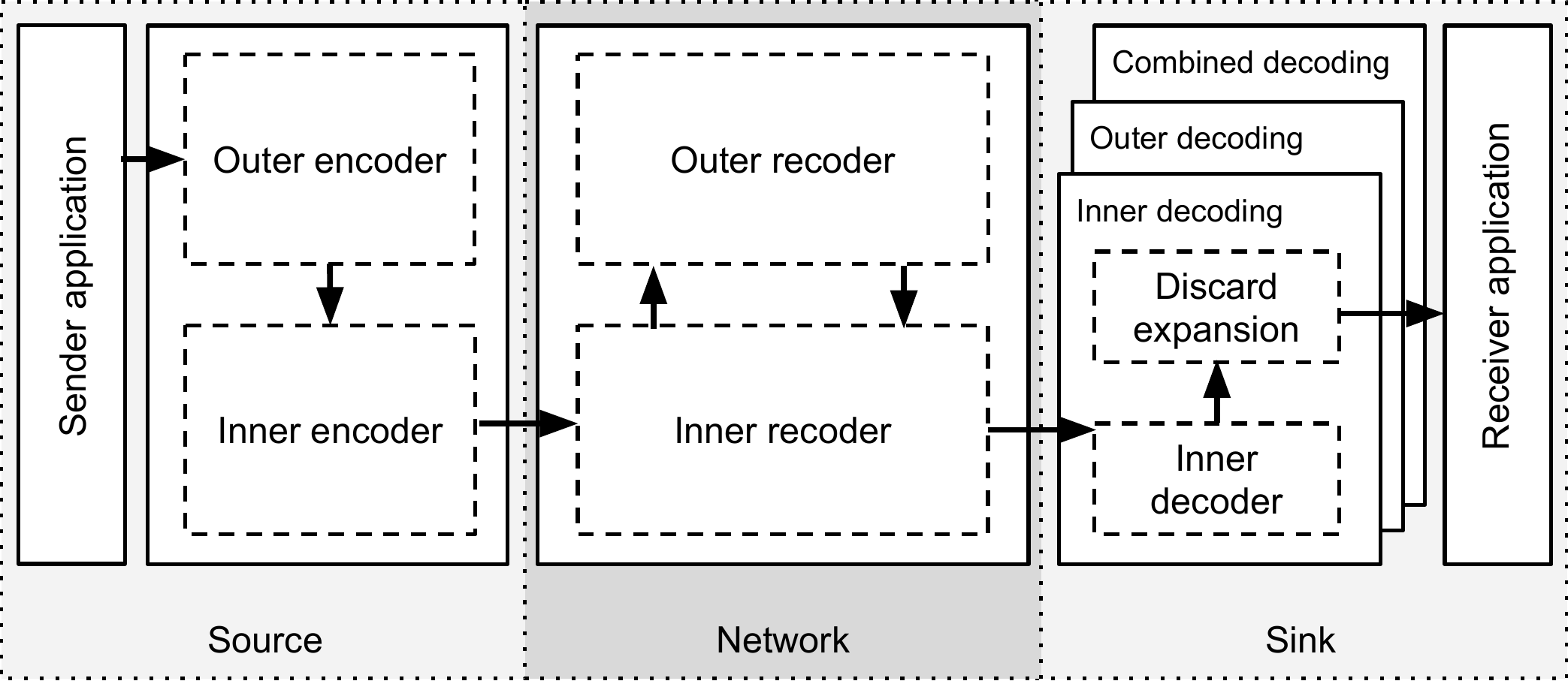}
\caption{ Description of system showing the inner and outer code
  structures. The outer code is typically established end-to-end. Although some
  applications could use outer recoders at intermediate nodes for higher efficiency
  in the network, in most scenarios the inner recoder is enough for
  supporting the desired functionalities. The sinks can choose from
  three main types of decoders: the inner, the outer, and the combined
  decoders. The outer can be exploited with any configuration of
  outer/inner codes, while the inner and combined decoders require a
  specific structure of the outer code, e.g., systematic.}
\label{fig:System}
\vspace{-0.3cm}
\end{figure*}

\vspace{-0.2cm}
\section{Description of the Scheme}\label{sec:description}
The key goals of Fulcrum network codes are the following:
\begin{enumerate}
\item Reduce the overall overhead of network coding by
  (a) reducing the overhead due to coding coefficients per packet, and
  (b) reducing the overhead due to transmission of linearly dependent
  packets.
\item Provide simple operations at the routers/devices in the
  network. The key is to make recoding at these devices as simple as
  possible, without compromising network coding capabilities.
\item Enable a simple and adaptive trade-off between performance and
  complexity.
\item Support compatibility with any end-to-end linear erasure code in
  $GF(2^h)$.
\item Control and choose desired performance and effort in end
  devices, while intermediate nodes provide a simple, compatible layer
  for a variety of applications.
\end{enumerate}

\subsection{Idea}
The key technical idea of Fulcrum is the use of a dimension expansion
step, which consists of taking a batch of $n$ packets, typically
called a generation, from the original file or stream and expand into
$n + r$ coded packets, where $r$ coded packets contain redundant
information and are called \emph{expansion packets}. After the
expansion, each resulting coded packet is treated as a new packet that
will be coded in $GF(2)$ and sent through the network.  The mapping
for this conversion is known at the senders and at the receivers
\textit{a priori}.

Since additions in any field of the type $GF(2^k)$ is simply a
bit-wise XOR, the underlying linear mapping in higher fields can be
reverted at the receivers. The reason to do the expansion is related
to the performance of $GF(2)$, which can introduce non-negligible
overhead in some settings~\cite{Nistor11,Lucani09FieldSize}. More
specifically, coded packets have a higher probability of being
linearly dependent when more data is available at the
receiver. Increasing dimensions addresses this problem by mapping back
to the high field representation after receiving $n$ linearly
independent coded packets and decoding before the probability of
receiving independent combinations in $GF(2)$ becomes prohibitive.
The number of additional dimensions, $r$, controls the decoding
probability performance. The larger the $r$, the better the
performance achieved by the receivers while still using $GF(2)$ in the
network.

Our approach naturally divides the problem in the design of inner and
outer codes, using the nomenclature of concatenated
codes~\cite{Forney66}. Concatenating codes is a common strategy in
coding theory, but typically used solely for increasing throughput
performance point-to-point~\cite{Forney66} or end-to-end, e.g., Raptor
codes~\cite{Shokrollahi06}. Some recent work on NC has considered the
idea of using concatenation to (i) create overlapping generations to
make the system more robust to time-dependent losses, but using the
same field size in the inner and outer code~\cite{Thibault08}; (ii)
decompose the network in smaller sub-networks in order to simplify
cooperative relaying~\cite{Kim07Concatenated}; (iii) connecting NC and
error correcting channel coding, e.g.,~\cite{Yin10}; or (iv) subspace
codes for noncoherent network coding~\cite{Skachek13}. Fulcrum is
fundamentally disruptive in two important ways. First, we allow the
outer code to be agreed upon by the sources and receivers (dimension
expansion), while the inner code is created in the network by recoding
packets. Thus, we provide a flexible code structure with controllable
throughput performance. Second, it provides a conversion from higher
field arithmetic to $GF(2)$ to reduce complexity.

Dividing into two separate codes has an added advantage, not
envisioned in previous approaches. This advantage comes from the fact
that the senders can control the outer code structure to accommodate
heterogeneous receivers. The simplest way to achieve this is by using
a systematic structure in the outer code. This provides the receivers
with the alternative to decode in $GF(2)$ after receiving $n + r$
coded packets instead of mapping back to higher fields after receiving
$n$ coded packets. This translates into less decoding complexity, as
$GF(2)$ requires simple operations, but incurring higher delay. The
latter comes from the fact that $r$ additional packets must be
received.

If the precoding uses a systematic structure, the system can support
three main types of receivers (See Fig.~\ref{fig:System}). First, a
computationally powerful receiver that decodes in $GF(2^h)$ by mapping
back from the $GF(2)$ combinations received. We call this the outer
decoder. This procedure is simple because the addition for any
extension field $GF(2^l)$ is the same as that to $GF(2)$, namely, a
bit-by-bit XOR.  We show that accumulating $n$ linearly independent
$GF(2)$ coded packets is enough to decode in the higher field.  A
receiver that decodes in $GF(2)$ reduces its decoding complexity but
needs to gather $n+r$ independent linear combinations. Finally, we
show that a hybrid decoder is possible, which can maintain the high
decoding probability when receiving $n$ coded packets as in the
high-field decoder, while having similar decoding complexity to that
of the inner decoder. We call this hybrid decoder the combined
decoder.

Our work is inspired in part by~\cite{Thomos12}, which attempted to
maintain overhead limited to a single symbol per packet. Thomos et al.
made a very careful design in their packet coding at the source, but
the end result is seemingly disappointing because only a small number
of packets could be transmitted maintaining the overhead at one
symbol. However, we argue that their careful code construction is not
really needed. In fact, the reason behind their results is dominated
by the network operations and their strict overhead limitation and not
the source code structure, as we will show in this paper. Through our
simple design framework, we break free from the constraint of a single
symbol overhead and discover the potential to (i) reduce the overhead
per packet in the network to roughly that of an end to end $GF(2)$
RLNC system (which is equivalent to the overhead reported
in~\cite{Thomos12}), (ii) trade-off performance in the presence of
heterogeneous receivers exploiting a family of precoders and simple
designs, and (iii) exploit any generation size without introducing a
synthetic constraint due to the field size at the precoder. The work
in~\cite{Thomos12} is a special subcase of our general framework.

\subsection{Design}
The overall framework is described in
Fig.~\ref{fig:System}, showing the actions of the source, the network,
and the destinations. In the following, we describe these actions in
more detail.

\subsubsection{Operations at the Source} using the $n$ original
packets, $P_1, P_2, ..., P_n$, the source generates $n + r$ coded
packets, $C_1, C_2, ..., C_{n+r}$ using $GF(2^h)$ operations (See
Fig.~\ref{fig:System}-Source). The additional $r$ coded packets are
called \emph{expansion packets}.  After generating these coded
packets, the source re-labels these as mapped packets to be sent
through the network and assigns them binary coefficients in
preparation for the $GF(2)$ operations to be carried in the network.
Finally, the source can code these new, re-labeled packets using
$GF(2)$. The $i$-th coded packet has the form $\sum_{j=1}^{n+r}
\lambda_{i,j} C_j$. The coding over $GF(2)$ is performed in accordance
with the network's supported inner code. For example, if the network
supports $GF(2)$ RLNC, the source generates RLNC coded
packets. However, other inner codes, e.g., perpetual~\cite{VTC2014},
tunable sparse network coding~\cite{Feizi12, NCSmartGrid}, are also
supported.

Our main design constraint is that the receiver should (i) decode with
the $n+r$ coded packets, and more importantly (ii) that it can decode
with high probability after the reception of $n$ coded packets. Given
that the structure of the initial mapping is controlled by the source,
we could use Reed-Solomon (RS) codes, which are known end-to-end, or send
the seed that was used to generate the mapping with each packet for a
random code. In order to cater to the capabilities of heterogeneous
receivers, we suggest the use of a systematic expansion/mapping, which
already guarantees condition (i) but also provides interesting
advantages for computationally constrained receivers. This is
explained in more detail in the operations of the receivers.

\subsubsection{Operations at Intermediate Nodes}
The operations at the intermediate nodes are quite
simple (Fig.~\ref{fig:System}-Network). Essentially, they will receive coded packets in $GF(2)$ of the
form $\sum_{j=1}^{n+r} \lambda_{i,j} C_j$, store them in their
buffers, and send recoded versions to the next hops, typically implementing an inner
decoder as described in the following. The recoding mechanism is what
defines the structure of the inner code of our system.  Recoding can
be done as a standard $GF(2)$ RLNC system would do, i.e., each packet
in the buffer has a probability of $1/2$ to be XORed with the others
to generate the recoded packet. However, the network can also support
other recoding mechanisms, such as recoding for tunable sparse network
coding~\cite{Feizi12,NCSmartGrid} and for Perpetual network
codes~\cite{VTC2014}, or even no recoding. We shall discuss
the effect of several of these inner codes as part of
Section~\ref{sec:sparse-inner-codes}.

In some scenarios, it may be possible to allow intermediate nodes to
know and exploit the outer code in the network
(Fig.~\ref{fig:System}-Network). The main goal of these recoders is to
maintain recoding with $GF(2)$ operations only. However, when the
intermediate node gathers $n$ linearly independent coded packets in
the inner code, it can choose to map back to the higher field in order
to decode the data and improve the quality of the recoded packets. The
rationale is that, at that point, it can recreate the original code
structure and generate the additional dimensions $r$ that are missing
in the inner code, thus speeding up the transmission process. Although
not required for the operation of the system, this mechanism can be
quite useful if the network's intermediate nodes are allowed to
trade-off throughput performance with complexity.

\subsubsection{Operations at the Receivers} We consider three main
types of receivers assuming that we enforce a systematic outer
code. Of course, intermediate operations or other precoding approaches
can be used enabling different end-to-end capabilities and
requirements.

Receivers using an \textit{Outer Decoder} will map back to the
original linear combination in $GF(2^h)$ (See
Fig.~\ref{fig:overview_decoders}-left). This means that only decoding
of an $n \times n$ matrix in the original field is required. The
benefit is that the receiver decodes after receiving $n$ independent
coded packets in $GF(2)$ with high probability. The key condition is
that the receivers need to know the mapping in $GF(2^h)$ to map back
using the options described for the source. These receivers use more
complex operations for decoding packets, but are awarded with less
delay than $GF(2)$ to recover the necessary linear combinations to
decode. We will show that increasing $r$ yields an exponential
decrease of the overhead due to non-innovative packets, i.e., coded
packets that do not convey a new independent linear combination to the
receiver.

On the other hand, receivers using an \textit{Inner Decoder} opt to
decode using $GF(2)$ operations for the $n+r$ relabelled packets (See
Fig.~\ref{fig:overview_decoders}-middle). This is known to be a
faster, less expensive decoding mechanism although there is some
additional cost of decoding a $(n + r) \times (n+r)$ matrix. If the
original mapping uses a systematic structure, decoding in this form
already provides the original packets without additional decoding in
$GF(2^h)$. The penalty for this reduced computational effort is the
additional delay incurred by having to wait for $n+r$ independent
linear combinations in $GF(2)$. Thus, there is no benefit over
standard $GF(2)$ but we provide compatibility with the other nodes.

Finally, receivers using a \textit{Combined Decoder} implement a
hybrid between inner and outer decoders with the aim of approaching
the decoding speed as inner decoders while retaining the same decoding
probability as outer decoders (See
Fig.~\ref{fig:overview_decoders}-right). This is achieved by decoding
the first $n$ coded packets using $GF(2)$ only. If decoding is
unsuccessful in $GF(2)$, all coded packets are mapped to $GF(2^h)$
over which the remaining decoding is performed. Hence, if $r << n$ the
decoding cost of the last $r$ is negligible compared to that of the
initial $n$ packets and decoding speed will approach that of an inner
decoder. We show this in Section~\ref{sec:implementation}.

\begin{figure*}[t]
\centering
\includegraphics[width=0.7\textwidth]{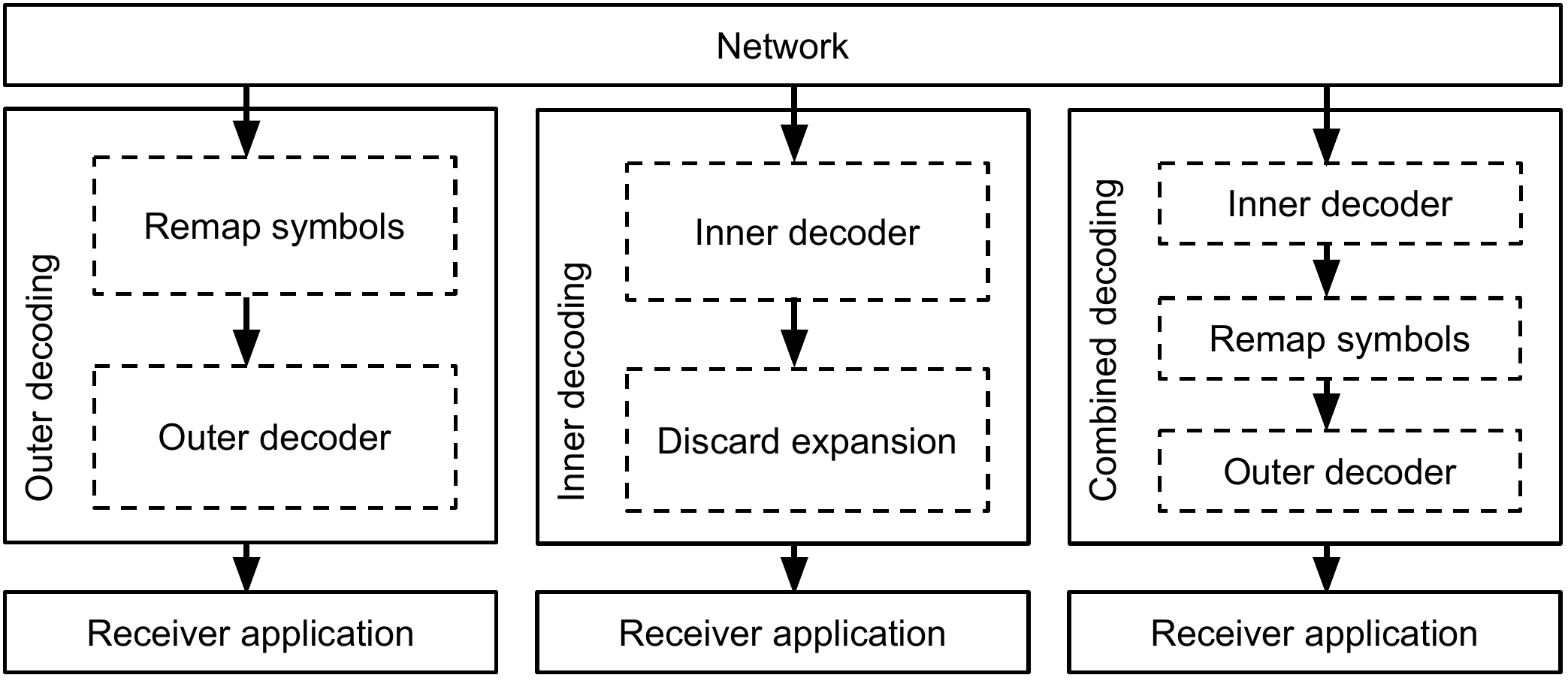}
\caption{ Overview of components of the three types of Fulcrum
  decoders.  Each decoder operates on the same data stream coming from
  the network. This makes it possible to support heterogeneous
  receivers using some mix of the three decoding types. }
\label{fig:overview_decoders}
\vspace{-0.5cm}
\end{figure*}


\subsection{Benefits}

\subsubsection{Simple is Green, Compatible, and Deployable}
The bit-wise XOR operations of $GF(2)$ in the network are easy and
cost-effective to implement in software/hardware and also energy
efficient because they require little processing. Their simplicity
makes them compatible with almost any device and can be processed at
high-speed. Having a simple approach where all packets in the network
are processed using $GF(2)$, reduces the additional logic and
processing in intermediate nodes.

\subsubsection{Supports Heterogeneous Receivers}
The use of an inner decoder is not only an option for resource-limited
devices. It can be a tool to reduce the decoding effort at receivers, in particular if $r$ is small. This is possible
if the outer code has a systematic structure, which allows for a
direct recovery of the original packets once decoding happens in
$GF(2)$, or other structure for the outer code, that allows for a fast
decoding, e.g., sparse outer code, Jacket matrix
structure~\cite{JacketM}. For example, consider the case of a single
hop broadcast network with heterogeneous channels. A device that has
the best channel of the receivers could choose to wait for additional
transmissions from the source, given that those transmissions will
happen in any case and the receiver might still invest energy to
receive them (e.g., to wait for the next generation). The additional
receptions can be used to decode using $GF(2)$ instead of doing the
mapping to decode in $GF(2^h)$ in this way reducing the energy
invested in decoding. On the other hand, a receiver with a bad channel
will attempt to use $GF(2^h)$ to decode in order to reduce its overall
reception time.

\subsubsection{Adaptable Performance}
Fulcrum can be configured to cover a wide range of decoding complexity
vs. throughput performance tradeoffs if we consider the use of sparse
outer codes. Let us illustrate this potential with a simple outer
code with a density $d$, i.e., the fraction of non-zero
coefficients. Let us consider some extreme cases. First, the case of
no extension $r=0$ and $d \approx \frac{n}{2}$, which corresponds to a
binary \ac{RLNC} code. Second, the case of no extension $r=0$ and $d
<< \frac{n}{2}$, which corresponds to a sparse binary \ac{RLNC}
code. Also, $r>>n$, $d = 1$, can achieve similar performance to
\ac{RLNC} with a high field size. By adjusting $r$ and $d$, Fulcrum
can be tuned to the desired complexity - throughput
tradeoff. Importantly, these parameters can also be changed on-the-fly, allowing for
adoption to time varying conditions. Of course, the choice of outer
code need not be only changed with the $d$ parameter. In fact, an LT
code, Perpetual code, or Tunable Sparse Network Code could also be
used for providing additional trade-offs.

\subsubsection{Practical Recoding}
Recoding can be performed exclusively over the inner code, while
encoding and decoding is performed over the outer code. As the inner
code is in $GF(2)$, its coding vector can easily be represented
compactly, which solves the challenge associated with enabling
recoding when a high field size is used~\cite{Heide11}.

\subsubsection{Spreading Complexity Across the Network}
Both intermediate nodes and receivers can choose the computational
complexity they are capable of dealing with. This enables the
complexity to be spread across nodes in the network.  Although we
currently focus on $n$ independent decision between receivers, future
work could consider a network that controls what devices will invest
more computational effort for specific flows.

\subsubsection{Security}
If security is the goal, our scheme provides a simple way to implement
some of the ideas in SPOC~\cite{Vilela08}. With Fulcrum, the mapping
of the outer decoder constitutes the secret key (or part of it) that
the source and destinations share and that, in contrast
to~\cite{Vilela08}, need not be sent over the network along with the
coded packets.  Using Fulcrum, we will not incur the large overhead
of SPOC, which sends two coding coefficients per original packet (one
encrypted, one without encryption). In fact, the end points (source
and receivers) can choose very large field sizes in the outer code
while maintaining $1 + r/n$ bits per packet in the generation as
overhead. Fulcrum can also provide security without the need to run
Gaussian elimination twice at the time of decoding~\cite{Zhang10}. As
a consequence Fulcrum does not need to trade--off field size and
generation size (and thus security) for overhead in the network and
complexity.

\subsubsection{New Designs while Supporting Backwards Compatibility}
Exploiting one code in the network and one underlying code end-to-end
provides senders and receivers with the flexibility to control their
service requirements while making the network agnostic to each flows'
characteristics. This has another benefit: new designs and services
can be incorporated with minimal or no effort from the network
operator and maintaining backward compatibility.

\vspace{-0.2cm}
\section{Analysis of Receiver Performance with RLNC as Inner
  Code}\label{sec:analys-rx-perf} Let us understand the delay
performance in receivers using the outer and combined
decoders. Receivers using an inner decoder correspond to a $GF(2)$
receiver that needs to get $n+r$ independent linear combinations
before decoding~\cite{Nistor11, Heide09, Lucani09FieldSize}. The key
question is to determine whether receiving $n$ independent coded
packets in $GF(2)$ means that the re-mapped version in $GF(2^h)$ is
full rank, i.e., original data can be
decoded. Section~\ref{sec:RSOuterCode} shows that this is possible for
a RS outer code under some minor conditions.

\subsection{Decoding Performance with a Reed-Solomon Outer Code}
\label{sec:RSOuterCode}

We have $B$, an $m \times n+r$-matrix over $\FF_2$, and $G$, an $(n
+r) \times n$-matrix over $\FF_{2^s} = \FF_q$. We have that $G$ is the
generator matrix of a RS code with length $n+r \leq q -1 = 2^s
-1$. The value of $m$ is related to the length of the incoming
messages, e.g., if it is a single $\FF_{2^s}$ symbol, then $m =
s$~bits. We remark that vectors are column vectors and that we
multiply on the right.

A RS code $C$, with dimension $n$, can be defined as the
vector space generated by the evaluation of the monomials $1, X,
\ldots, X^{n-1}$ at the points $\FF_q \setminus \{0\}$. Namely, let
$\alpha$ be a primitive element of $\FF_q$ and let $\ev: \FF_q [X] \to
\FF_q^{n+r}$, given by $\ev(f) = (f(\alpha^0), f(\alpha^1), \ldots,
f(\alpha^{q-2})$. One has that $C = \langle \{ \ev (X^i) : i = 0,
\ldots , n-1 \} \rangle $. A generator matrix is given by considering
as columns the evaluation of a monomial at $\FF_q \setminus \{0\}$.
The dual code of a RS code is given by Lemma~\ref{lem:dual}.

\begin{lem}\label{lem:dual}
  (Lemma 5.3.1~\cite{EMS}) Let $C$ be a Reed-Solomon code with dimension $n$, then the dual
  code of $C$ is given by $C^\perp = \langle \{ \ev (X^1) , \ldots ,
  \ev (X^{q-1-n}) \} \rangle$
\end{lem}


We consider the $m \times n$-matrix $BG$ and denote the associated
linear function by $\varphi$. We assume that $B$ and $G$ have full
rank and we wonder whether $\dim (\varphi (V)) < \dim V$ for a vector
subspace $V \subseteq \FF_q^n$. Since $\varphi$ is a linear function,
one has that the dimension of the image plus the dimension of the
kernel is equal to the dimension of the original space. Therefore, we
wonder whether, $\dim (\ker (\varphi) ) > 0$.

In order to prove the main result, we shall introduce the cyclotomic
coset containing $a$ in $\FF_q = \FF_{2^s}$, $I_a = \{a , 2a \modulo
q-1, 2^2 a \modulo q-1, \ldots , 2^{s-1} a \modulo q-1 \}$. For
instance, for $q= 2^4$, the different cyclotomic cosets
are $I_0=\{0\}, I_1=\{1,2,4,8\}, I_3=\{3,6,12,9\}, I_5=\{5,10\},
I_7=\{7,14,13,11\}.$ One has that $I_1=I_2=I_4=I_8$, but usually one
denotes the coset by the smallest number.
We can now characterize when  $\dim (\ker (\varphi) ) = 0$.

\begin{thm}\label{thm:RScondition}
  Let $C \subseteq \FF_{2^s} = \FF_q$ be a Reed-Solomon code with
  dimension $n$ and the linear map $\varphi$ defined above. One has
  that $\dim (\ker (\varphi) ) = 0$ if and only if $n \ge 2^{s-1}$.
\end{thm}

\begin{proof}
  One has that $\varphi$ is the composition of two linear maps, the
  ones associated to $G$ and $B$. We have that $G$ is a generator
  matrix and therefore, it is injective. Hence, $BG x = 0 $ if and
  only if $ c=Gx \in \ker (B)$. That is, $\dim (\ker (\varphi) ) > 0$
  if the rows of $B$ are orthogonal to $c$, which is a word of
  $C$. Therefore, the rows of $B$ are words in the dual code of $C$.

  By Lemma \ref{lem:dual}, the dual code of $C$ is given by $C^\perp =
  \langle \{ \ev (X^1) , \ldots , \ev (X^{q-1-n}) \} \rangle$. We have
  that $C^\perp \subseteq \FF_q^{n+r}$ but the rows of $B$ are over
  $\FF_2$. Hence we should consider the subfield subcode of $C^\perp$,
  that is
$$
\sub (C^\perp) = \{c \in \FF_2^{n+r}:  c \in C^\perp \}.
$$The columns of the generator matrix of $\sub(C^\perp)$ are the rows of $B$ which reduce the dimension of the image.
By \cite{MS} and \cite[Theorem III.8]{Fer}, one has that $$ \dim (\sub
(C^\perp)) = \# \{ I_j : I_j \subseteq \{1,\ldots , q-1-n\} \}.
$$
That is, we shall only consider exponents that are in a cyclotomic
coset that is contained in $\{1,\ldots , q-1-n\}$. Clearly $I_0 \not
\subseteq \{ 1, \ldots , q-1-n \}$. Let $k \ge 2^{s-1}$, then one has
that $q-1-n < 2^{s-1}$ and therefore the cyclotomic coset $I_1 =
\{1,2,2^2, \ldots, 2^{s-1}\}$ is not contained in $\{ 1, \ldots , q-1-n
\}$. Finally, let $j > 2^{s-1}$, then $I_j \neq I_0, I_1$, and we have
that $I_j \not \subseteq \{ 1, \ldots , q-1-n \}$ since $j \in I_j$
and $q-1-n < 2^{s-1}$. Therefore $\sub (C^\perp) = \{ 0 \}$ and $\dim
(\ker (\varphi) ) = 0$.

Let us consider now $n < 2^{s-1}$, then $I_1 = \{1,2,2^2, \ldots,
2^{s-1}\} $ is contained in $\{ 1, \ldots , q-1-n \}$ and $\dim (\sub
(C^\perp)) \ge s$. Thus $\dim (\ker (\varphi) ) >0$.
 \end{proof}


Explicit generators of $\sub (C^\perp)$ to identify cases of linear
dependence for $n < 2^{s-1}$ can be obtained by using results in
\cite{MS}, \cite[Theorem III.8]{Fer} as follows.
\begin{thm}
 Let $C$ be a Reed-Solomon code with dimension $n$, the dimension of $\sub (C^\perp)$ is
$\sum_{I_a \subseteq \{1,\ldots q-1-n\}} \# I_a, $
and a basis is given by $\{  \ev (f_{I_a,\beta^j}) : j \in \{0, \ldots, \#I_a -1\}  \},
$ where $\beta=\alpha^{(2^s-1)/(2^{\# I_a} -1)}$, i.e. a primitive element of $\FF_{2^{\# I_a}} \subseteq \FF_{2^s}$, and $f_{I_a,\beta} = \beta X^a + \beta^2 X^{2a} + \cdots + \beta^{2^{\#I_a -1}} X^{2^{\#I_a -1}a}$.
\end{thm}
As an example, let $C_8$ be the Reed-Solomon code with dimension
$n=2^3$ in $\FF^{15}_{2^4}$. We have that $C_8 = \langle \{ \ev
(X^0) , \ldots , \ev (X^7) \} \rangle $, $C_8^\bot = \langle \{ \ev
(X^1) , \ldots , \ev (X^7) \} \rangle $ and $\sub (C_8^\perp) = \{0
\}$ because $I_a \not \subseteq \{1 , \ldots , 7\}$ for any $a$ and
$\dim (\ker (\varphi) )= 0$.  
Consider another example with a Reed-Solomon code $C_7$
with dimension $n=7$ in $\FF^{15}_{2^4}$. We have that $C_7 = \langle
\{\ev (X^0) , \ldots , \ev (X^6) \} \rangle $, $C_7^\bot = \langle \{
\ev (X^1) , \ldots , \ev (X^8) \} \rangle $ and $\dim (\sub
(C_7^\perp)) = 4$ because $I_1 = \{1,2,4,8 \} \subseteq \{1 , \ldots ,
8\}$. 
A basis for $\sub (C_7^\perp)$ is given by $\{ \ev (f_{I_1,1}) ,
\ev (f_{I_1,\alpha}) , \ev (f_{I_1,\alpha^2}), \ev (f_{I_1,\alpha^3})
\}$, that is $$\{ \ev(X + X^2 + X^4 + X^8) , \ev(\alpha X + \alpha^2
X^2 + \alpha^4 X^4 + \alpha^8 X^8),$$ $$ \ev(\alpha^2 X + \alpha^4 X^2
+ \alpha^8 X^4 + \alpha X^8), \ev(\alpha^3 X + \alpha^6 X^2 +
\alpha^{12} X^4 + \alpha^9 X^8) \}.$$ Therefore, a generator matrix for
$\sub (C_7^\perp)$ is
$$\colvec{0 0 0 1 0 0 1 1 0 1 0 1 1 1 1\\
0 0 1 0 0 1 1 0 1 0 1 1 1 1 0\\
0 1 0 0 1 1 0 1 0 1 1 1 1 0 0\\
1 0 1 1 1 0 1 0 1 1 1 1 0 0 0}^T.$$

Future work shall consider exploiting such generators to improve the
efficiency of the decoder in these corner cases.  In order to consider
$\beta$ as a primitive element of $\FF_{2^{\# I_a}}$ in $\FF_{2^s}$,
one may consider to use Conway polynomials for defining finite fields
\cite{Con}.



\subsection{Delay Modelling and Performance} 
We focus on the analysis for a single receiver first to understand the
potential.  For the analysis, let us assume that a RS code
over $GF(2^s)$ and $n \geq 2^{s-1}$ is used for the outer code, so
that receiving $n$ independent coded packets in $GF(2)$ guarantees
that the re-mapped version in $GF(2^h)$ can be decoded (using the
result from Theorem~\ref{thm:RScondition}). We reuse the models for
RLNC coding from~\cite{Lucani09FieldSize,
  Nistor11}. Fig.~\ref{fig:MarkovChain} (a) shows the Markov chain
representing the process of reception of independent linear
combinations at the receiver given the reception of a new coded packet
over $GF(2)$ using an RLNC inner code. Each stage represents the
missing independent linear combinations in $GF(2)$ in order to decode
using only $GF(2)$ operations. In this case, we assume that the
receiver is attempting to decode in $GF(2)$ even when the source has
made an expansion to $n+r$ dimensions. This corresponds to a receiver
using an inner decoder. Fig.~\ref{fig:MarkovChain} (b) on the other
hand shows the process for a successful outer (and combined)
decoder. In this case, the underlying $GF(2)$ process needs only run
until $n$ independent linear combinations in $GF(2)$ are received,
which are mapped back into the $GF(2^h)$ and decoded for the outer
decoder. The combined decoder performs partial decoding in $GF(2)$
before attempting to use the high field, but this does not affect the
following analysis, only decoding complexity. Thus, the last states
starting in $r-1$ until $0$ are not visited, as state $r$ became an
absorbing state. If a different precoding structure is used, there
will be some probability of visiting the states beyond $r$. However,
if we use a large enough field size, this effect will be negligible
and the process described in Fig.~\ref{fig:MarkovChain} (b) will be a
very good approximation of the expected performance.  Using the
intuition from Fig.~\ref{fig:MarkovChain} (b), the mean number of
packets received from the network to decode using an outer (or
combined) decoder considering $r$ additional dimensions is given by
\begin{equation} \label{eq:MeanPacketsType1}
E\left[ N_{GF(2)} (r)\right] = \sum_{i=r+1} ^{n+r} \frac{1}{1-2^{-i}} = n +\sum_{i=r+1} ^{n+r} \frac{1}{2^{i} -1}.
\end{equation}

\begin{figure}[t]
\centering
\includegraphics[width=0.45\textwidth]{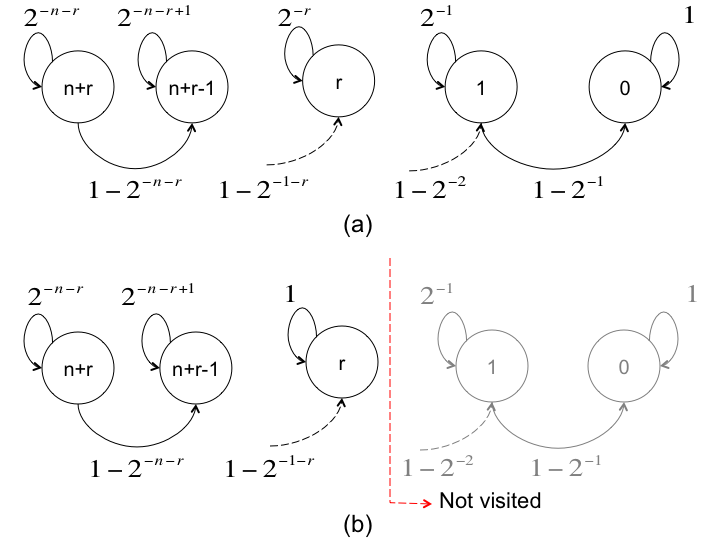}
\caption{ Markov Chain describing the reception process at a receiver
  in (a) a classical $GF(2)$ network, i.e., where both source and
  network use $GF(2)$ operations, and (b) exploiting our expansion
  idea, i.e., where the network treats expanded packets as $GF(2)$
  packets while the source uses $GF(2^h)$ operations for the
  expansion.}
\label{fig:MarkovChain}
\vspace{-0.3cm}
\end{figure} 

Lemma~\ref{lem:OverheadType1} shows that the overhead due to
additional $GF(2)$ coded packet receptions when using an outer
or combined decoder decreases exponentially with $r$.
\begin{lem} \label{lem:OverheadType1}
Considering an outer or a combined decoder with an MDS inner code, we have
$E\left[ N_{GF(2)} (r)\right] = n + 2^{-r} \times \theta(n)$,
for some $\theta(n) \in [  1-2^{-n} , 2 - 2^{-n+1}  ]$.
\end{lem}
\begin{proof}
  The proof follows from finding an upper and lower bound on $E\left[
    N_{GF(2)}(r)\right]$ described in Eq.~\eqref{eq:MeanPacketsType1}.
  To derive the upper bound, we use the fact that $2^{i-1} \leq 2^i
  -1$ for $i\geq 1$ to convert into the sum of a set of elements of
  the geometric series. Thus,
$E\left[ N_{GF(2)}(r)\right] \leq  n +\sum_{i=r+1} ^{n+r} 2^{-i+1} = n + 2^{-r+1} - 2^{-n-r+1}. \notag
$
The lower bound follows a similar argument, but using the fact that
$2^i \geq 2^i -1$ for $i\geq 1$. Thus,
$E\left[ N_{GF(2)}(r)\right] \geq  n +\sum_{i=r+1} ^{n+r} 2^{-i} = n + 2^{-r} - 2^{-n-r}. \notag
$
\end{proof}

Another interesting result for receivers with outer and combined
decoders is shown in Lemma~\ref{lem:OverheadType1Var}, where we proof
that the variance of $N_{GF(2)} (r)$ decreases exponentially with
$r$.
\begin{lem} \label{lem:OverheadType1Var} Considering a receiver using
  an outer or combined decoder with an MDS inner code, then
$var\left( N_{GF(2)} (r)\right) =  O(2^{-r}).$
\end{lem}
\begin{proof}
  The proof follows by bounding the variance of $N_{GF(2)}(r)$.
  Defining $P_i = 1 - 2^{-n-r+i-1}$ and using independence in the
  Markov chain, it is straightforward to prove that
  $var\left( N_{GF(2)} (r)\right) = \sum_{i = 1}^{n}
  \frac{1-P_i}{P_i^2} $. 
  After some manipulations and
  using the fact that $2^{-r-1} \geq 2^{-n -r + i -1}$ for $i= 1, ...,
  n$ then,
\begin{equation}
  var\left( N_{GF(2)} (r)\right) =\sum_{i = 1}^{n} \frac{1}{\left( 1 - 2^{-n-r+i-1}\right)^2} - E\left[ N_{GF(2)}(r)\right]\leq \frac{E\left[ N_{GF(2)}(r)\right]}{2^{r+1} -1 } \leq \frac{n + 2^{-r+1}}{2^{r+1} -1},
\end{equation}
which concludes the proof.
\end{proof}

\begin{figure}
  \centering
\subfloat[ \label{fig:CDF}]{{\includegraphics[width=0.45\linewidth]{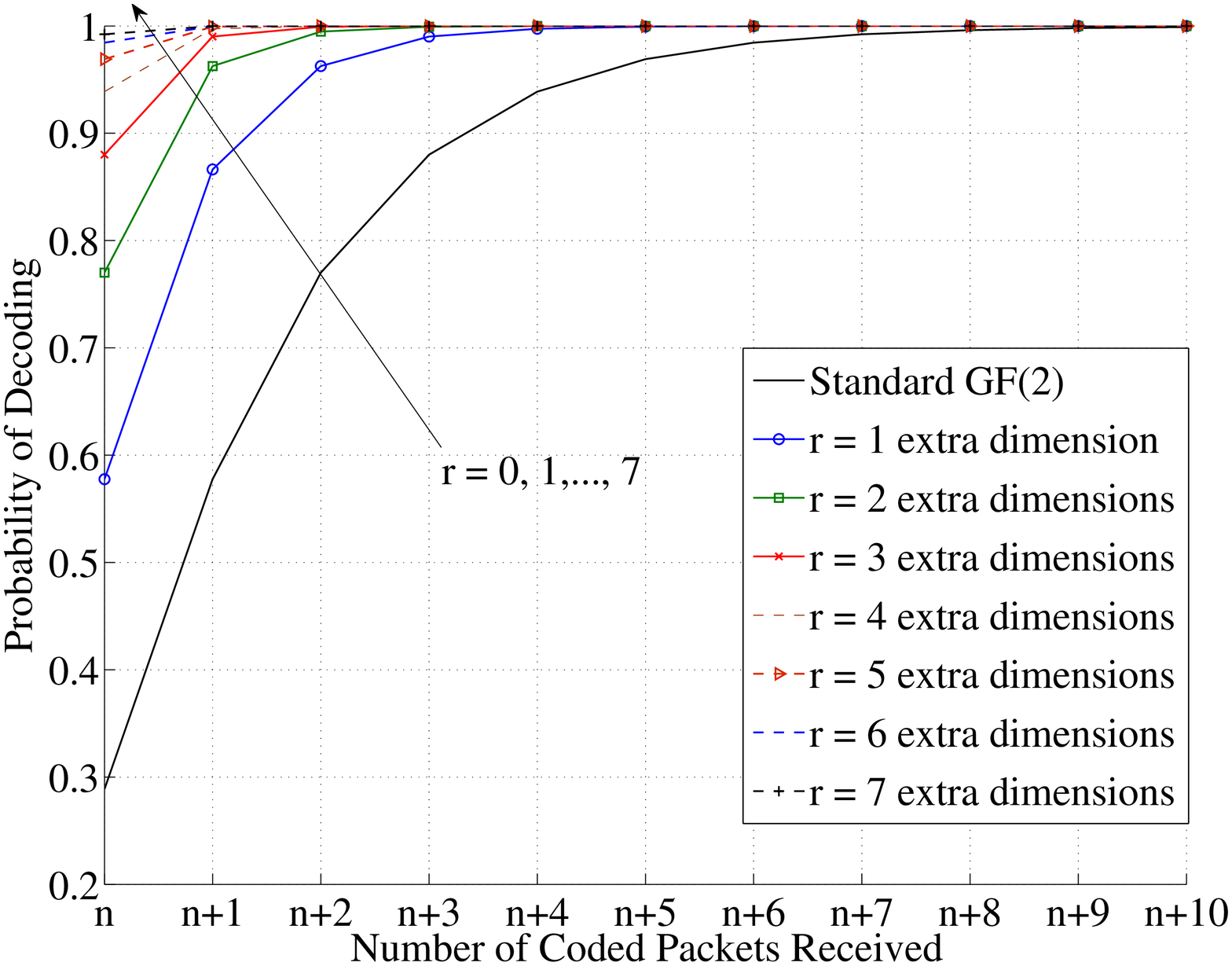}
 }}%
\qquad
    \subfloat[\label{fig:PMF} ]{{\includegraphics[width=0.45\linewidth]{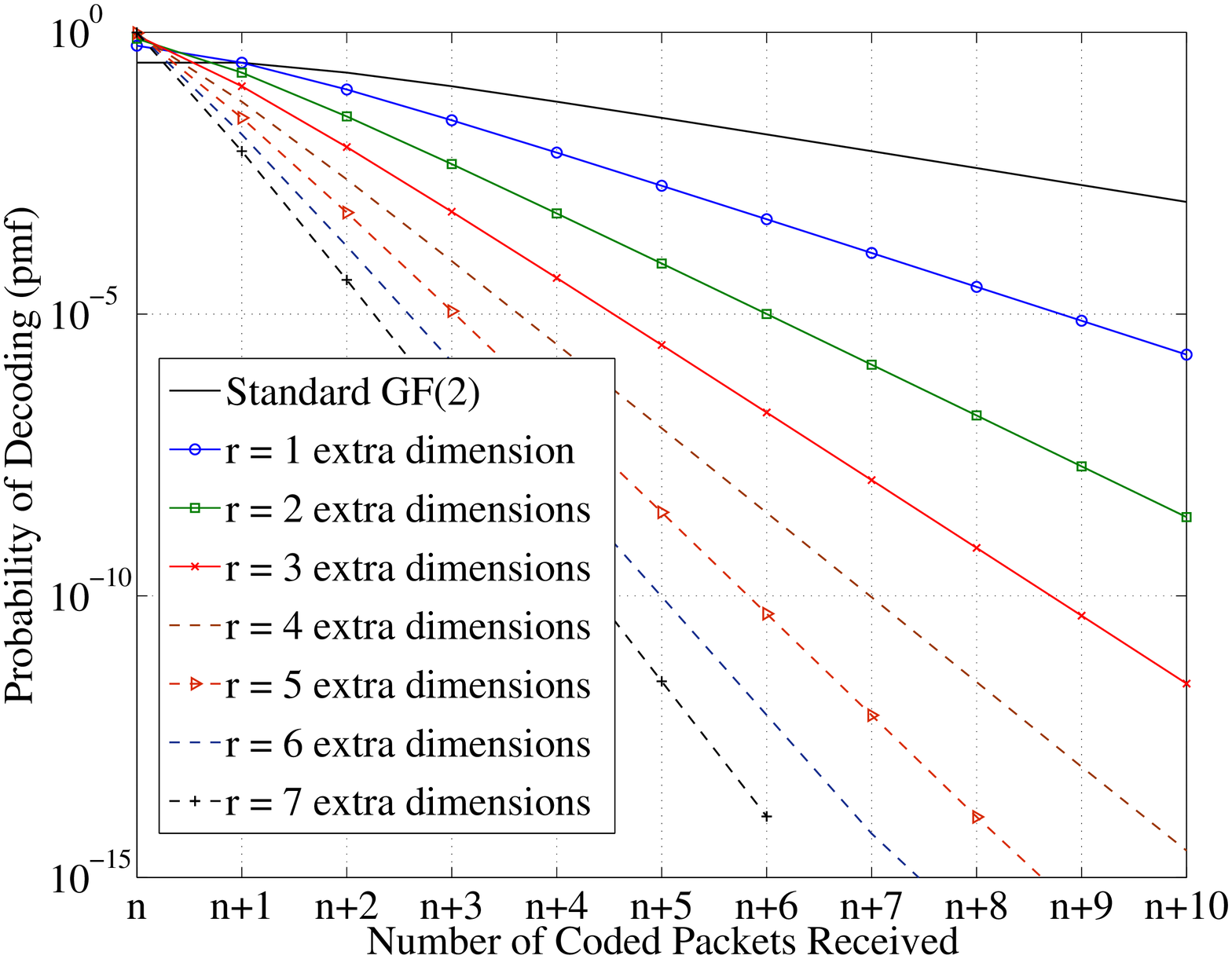}
      }}%
    \caption{Successful decoding for an outer or combined decoder
  given $r$ redundant dimensions: (a) CDF  and (b) PMF }
\vspace{-0.3cm}
\end{figure}



Fig.~\ref{fig:CDF} shows the cumulative distribution function (CDF)
for a receiver with an outer or a combined decoder and with a network
operating with $GF(2)$ and various values of $r$. Clearly, introducing
additional dimensions improves the probability of decoding with fewer
received coded packets from the network. In fact, even for a small
value, e.g., $r = 1$ or $r = 2$, the improvement is quite
noticeable. This is good from a practical perspective as these gains
can be achieved without dramatically decreasing performance in the
receivers using inner decoders. Table~\ref{tab:dec_prob} provides key
decoding probabilities (in percentages) when receiving $n$, $n+1$,
$n+2$, and $n+3$ coded packets when the inner code is RLNC
$GF(2)$. Better results could be obtained if using a systematic
structure. The table shows that the probability of decoding after
receiving exactly $n$ coded packets using an outer or combined decoder
is quite high even for moderate $r$ values. It also shows that the
performance with $r =7$ is similar to that provided by RaptorQ
codes~\cite{RaptorQ}, while $r >7$ can provide higher decoding
guarantees.  The corresponding probability mass function (PMF) is
presented in Figure~\ref{fig:PMF}. This shows that the variance is
reduced with the increase of $r$ and that even a
$r > 4$ reduces the probability of transmitting more than $n$ coded
packets before decoding by at least an order of magnitude with respect
to RLNC over $GF(2)$.

\begin{table*}[t]
\scriptsize
\center
\begin{tabular}{lrllll}
\toprule
 \multicolumn{2}{c}{\textbf{Code}}  & \multicolumn{4}{ c }{\textbf{Decoding after receiving (coded packets) }} \\
  \cmidrule(r){1-2}  \cmidrule(l){3-6}
   \multirow{4}{*}{\textbf{Fulcrum}} & $\mathbf{r}$ & $\mathbf{n}$ & $\mathbf{n+1}$ & $\mathbf{n+2}$ & $\mathbf{n+3}$ \\
   \cmidrule(lr){2-2} \cmidrule(l){3-6}
   & 4   & 93.87\% & 99.75\% & 99.99\% & 99.9997\% \\
   & 7  & 99.22\% & 99.996\% & 99.99998 \% & 99.99999992\% \\
  & 10 & 99.90\% & 99.9999\% & 99.99999996\% & 99.99999999998\%\\
  \cmidrule(r){1-2} \cmidrule(l){3-6}
  \textbf{RaptorQ} &  & 99\% & 99.99\% & 99.9999\% & \\
  \bottomrule
\end{tabular}
\caption{Decoding after reception of a certain number of coded packets
  using the outer or combined decoders for various $r$ and
  assuming RLNC $GF(2)$ inner encoder and recoders.}
\label{tab:dec_prob}
\end{table*}

\subsection{Extension to Broadcast with Heterogeneous
  Receivers} \label{sec:broadcast-results} Let us consider the case of
broadcast from one source to two receivers ($R_1$ and $R_2$) with
independent channels and packet loss probability $e_i$ for receiver
$R_i$.  Our goal is to illustrate the effect of using different
decoders at receivers with heterogeneous channel qualities as well as
to compare the performance of Fulcrum to that of standard RLNC at
different finite fields.

We exploit the Markov chain model presented in~\cite{Nistor11} to
provide an accurate representation of the field size effect when
broadcasting to two receivers. This model is also easily adapted to
incorporate the use of the outer decoding capabilities of Fulcrum.
The model in~\cite{Nistor11} relies on a state definition that
incorporates three variables, the number of independent linear
combinations at each receiver and the common linear combinations
between the two. The key change in the model is similar to the change
introduced in the Markov chain in Section~\ref{sec:analys-rx-perf},
that is, considering that the dimensions in the Markov chain
in~\cite{Nistor11} have a higher number of possible values, namely,
$n+r + 1$ instead of $n + 1$ per variable in the state. Then, if one
(or both) receivers use the outer (or combined) decoder, the number of
linear combinations gathered by that receiver will be increased to
$n+r$ whenever the receiver would normally achieve $n$.  Using these
modifications, we generated the results for
Figure~\ref{fig:CDF_Broadcast}.


\begin{figure}
  \centering
\subfloat[ \label{fig:CDF_Broadcast}]{{\includegraphics[width=0.45\linewidth]{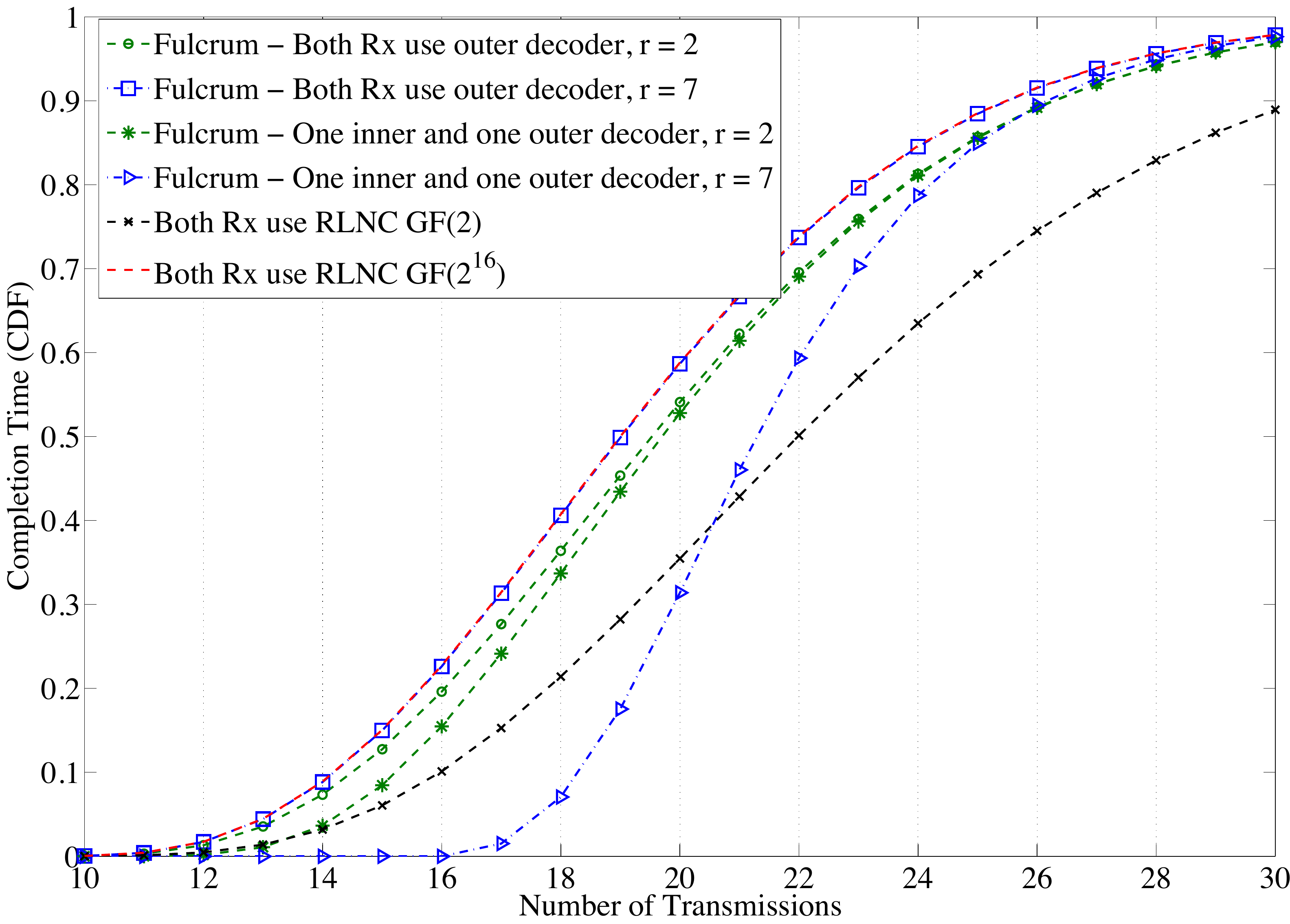}
 }}%
\qquad
\subfloat[ \label{fig:PMF_Broadcast}]{{\includegraphics[width=0.45\linewidth]{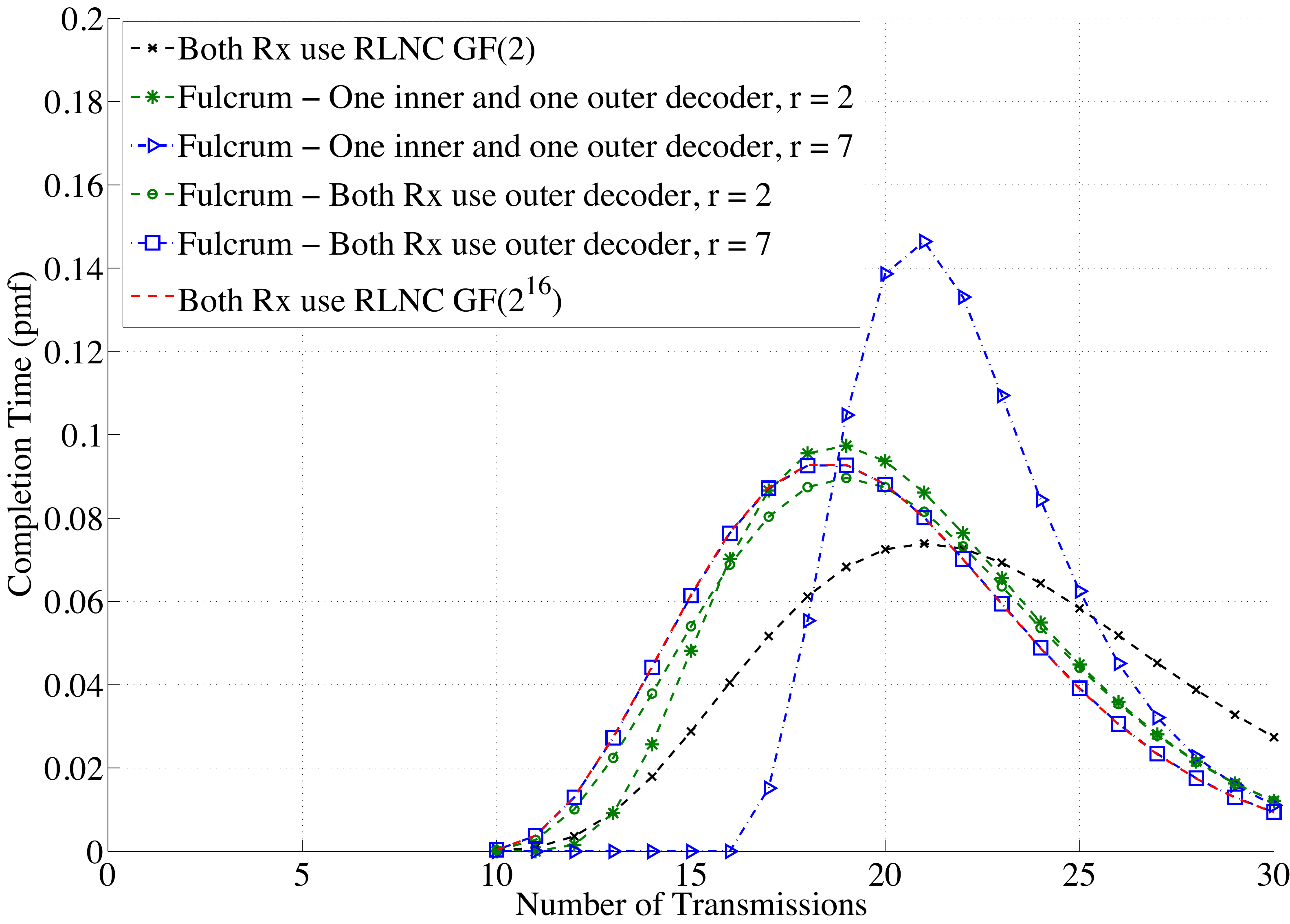}
 }}%
\caption{(a) CDF and (b) PMF of successful Decoding for Broadcast
  channel with two receivers showing the performance of standard RLNC
  for $GF(2)$ and $GF(2^{16})$, the performance of Fulcrum when one
  receiver exploits the outer code and the other only the inner code,
  and Fulcrum when both receivers exploit the outer code. Parameters
  $e_1 = 0.1$, $e_2 = 0.5$, $n = 10$ packets. }
\vspace{-0.3cm}
\end{figure}

Figure~\ref{fig:CDF_Broadcast} shows the CDF for the number of
transmissions to complete $n=10$ packets to two receivers. The
performance of RLNC $GF(2^{16})$ shows the best performance of the
best solution. Clearly, a Fulcrum approach where both receivers
exploit the outer decoder and use $r= 7$ redundant packets performs
essentially the same as RLNC with $GF(2^{16})$.  Note that Fulcrum
with $r=2$ and two receivers with outer decoder provides close to the
same performance as RLNC with $GF(2^{16})$ but with less
complexity. Additionally, $r = 2$ provides a much better trade--off
for the two receivers using different decoders.  Also, interesting is
that even when one of the receivers attempts to decode using the inner
decoder, i.e., using only $GF(2)$ operations, the worst case behavior
is maintained. This worst case behavior is superior to using only RLNC
with $GF(2)$.  Figure~\ref{fig:PMF_Broadcast} provides the PMF for a
similar scenario, demonstrating that the variance using Fulcrum is
reduced, particularly when increasing $r$. In fact, we observe that
for the case of two receivers using Fulcrum's outer (or combined)
decoder achieve essentially the same PMF as RLNC with $GF(2^{16})$.

\subsection{Overhead}
Let us define the \textit{overhead} for a generation with size $n$ as
the number of additional bits transmitted to successfully deliver a
generation of packets including coding coefficients and linearly
dependent packets. This means that it incorporates both the additional
header information and the overhead caused by retransmissions due to
linear dependency. For our analysis, we consider receivers with outer
and combined decoders and that we use the standard coding vector
representation, i.e., a coefficient per packet is sent attached to the
coded packet.

The mean overhead of using $GF(2^h)$ is proportional to $h n^2$~bits
for large enough $h$ considering that there are channel losses that
affect the total number of packets transmitted. The mean overhead of
our scheme is proportional to $E\left[ N_{GF(2)}\right] ( n + r ) \leq
n^2 + nr + (n+r) \Big( 2^{-r} - 2^{-n-r} \Big)$ if we use a standard
coding vector representation. Since $r << n$ for these cases, the
overhead will be dominated by $n^2 + n 2^{-r}$, which is to say almost
a factor of $h$ smaller than using $GF(2^h)$. However, this can be
further reduced if we use a sparse $GF(2)$ inner code.

\vspace{-0.2cm}
\section{Effect of Sparse Inner Codes} \label{sec:sparse-inner-codes}
Beyond benefiting from sparse inner codes for reducing overhead,
Fulcrum can provide fundamental benefits to sparse coding strategies
as well.
We consider the general case of static sparse structures, i.e., sparse
inner codes that do not change in time. Results can be extended to
time-variant schemes such as tunable sparse network
coding~\cite{Feizi12}.

For the general case of sparse matrices, we use the following bound.
\begin{thm}\label{thm:p-to-rho}
  (Theorem 1 of~\cite{TheoremProof1} ) The probability $p(i,n)$ of a
  received coded packet with density $\rho(i,n) \leq 1/2$ to be
  innovative when the receiver already has $i$ out of $n$ degrees of
  freedom is 
$p(i,n)\geq 1-(1-\rho(i,n))^{n-i}.
$
\end{thm}

If the number of non-zero coefficients is given by $k$, then the use
of Fulcrum provides a $\rho =
k/(n+r)$. Lemma~\ref{lem:SparseNumberOfNonZero} provides an upper
bound for the general case of sparse inner codes.

\begin{lem} \label{lem:SparseNumberOfNonZero}
Considering the problem of a sparse inner code with an MDS outer code,
we have
$E[S(n,r)] \leq n + n \left(\exp \left (  \frac{k (r+1)}{n+r} \right) - 1\right).
$
\end{lem}
\begin{proof}
\begin{align}
E[S(n,r)] & \overset{(1)}{\leq} \sum_{i = 0}^{n-1}
\frac{1}{1-(1-\rho)^{n+r-i}} 
= \sum_{i = r+1}^{n+r} \frac{1}{1 - \left( 1-k/(n+r)\right)^{i}}\\
&\overset{(2)}{\leq}  \sum_{i = r+1}^{n+r} \frac{1}{1 - \exp\left ( - \frac{ik}{n+r} \right)}
\end{align}
where step $(1)$ uses the bound in Theorem~\ref{thm:p-to-rho}, $(2)$
uses the fact that $\left( 1 - k/(n+r)\right)^i \leq \exp\left ( -
  \frac{ik}{n+r} \right)$.
Let us consider $q = \exp \left (
  \frac{k}{n+r} \right)$, then
\begin{align}
E[S(n,r)]&\leq \sum_{i=r+1}^{n+r} \frac{1}{1 - q^{-i}}
\overset{(3)}{\leq} n \frac{q^{r+1}}{q^{r+1} -1}
= n \frac{\exp \left (  \frac{k (r+1)}{n+r} \right)}{\exp \left (
    \frac{k (r+1)}{n+r} \right) - 1} 
\end{align}
where $(3)$ uses the fact that $q^{-i} \leq q^{-r-1}, q > 1, i \geq
r+1$, which concludes the proof.
\end{proof}

A first conclusion, is that
$\lim_{r\rightarrow \infty} E[S(n,r)] \leq  n \frac{e^k}{e^k- 1} = n \left( 1  + \frac{1}{e^k- 1}\right).
$
This means that adding additional redundancy in the outer code leads
to a bounded performance in terms of overhead. If $k = 3$ or $k=4$ for
a large outer code redundancy, the mean number of coded packets to be
received in order to decode is below $1.05 n$ and below $1.019 n$,
respectively. That is, less than $5$\% and $2$\% overhead.  If we
consider LT codes using an ideal soliton distribution, the mean number
of non-zeros is given by $k \approx \ln (n+r)$~\cite{MacKay02}, and $r
= \alpha n$ and $\alpha>0$, then
$E[S(n,r)] \leq  n + O\left ( \frac{n^{1/ (\alpha+1)} }{ (1+\alpha)^{ \alpha/(\alpha + 1)} } \right)
$. 
 As an example, if $\alpha = 1$, then $E[S(n,r)] \leq n \left(1 +
  \frac{1}{ \sqrt{2n} - 1} \right)$ while $\alpha = 2$ produces
$E[S(n,r)] \leq n \left(1 + \frac{1}{ (3n)^{2/3} - 1} \right)$.  If
$\alpha \rightarrow \infty$, the overhead vanishes.

For the special case of fixed $\rho$ irrespective of the number of
coding symbols for generated in the outer code, i.e., $\rho(i,n) =
\rho$, then
\begin{align}
E[S(n,r)] &\overset{(1)}{\leq} \sum_{i = r+1}^{n+r} \frac{1}{1 - \left( 1/ (1-\rho)\right)^{-i}}
 \overset{(2)}{=}  \sum_{i = r+1}^{n+r} \frac{1}{-B^{-i}} =  \sum_{i = r+1}^{n+r} \frac{B^i}{B^i -1}\\
& \overset{(3)}{\leq} n +  \sum_{i = r+1}^{n+r} \frac{1}{B^{i-1}
  (B-1)} = n + \frac{B}{B-1}  \sum_{i = r+1}^{n+r}  B^{-i}\\
&= n + \frac{B^{n+1} -B}{\left( B-1\right)^2 B^{n+r}} = n + \frac{(1-\rho)^{r+1} }{ \rho^2} - \frac{(1-\rho)^{n+ r+1} }{ \rho^2},
\end{align}
where step $(1)$ uses the bound in Theorem~\ref{thm:p-to-rho}, $(2)$
considers making $B = \frac{1}{1-\rho}$, and $(3)$ uses $B^i \geq B^i
- B^{i-1}$ for $B \geq 1$.  For $\rho$ fixed, then $r \rightarrow
\infty$ then $E[S(n,r)]$ tends to $n$.

\vspace{-0.25cm}
\section{Implementation}\label{sec:implementation}
In this section, we describe the implementation of different Fulcrum
encoder and decoder variants. The descriptions presented here are
based on our actual implementation of the algorithms in the Kodo
network coding library~\cite{kodogit}. For our initial implementation,
we utilized two RLNC codes with the outer code operating in $GF(2^8)$
or $GF(2^{16})$ and the inner code in $GF(2)$.

\subsection{Implementation of the Encoder}
One advantage of Fulcrum is that the encoding is quite
simple. Essentially, the two encoders can be
implemented independently, where the outer encoder uses the $n$ original
source symbols to produce $n+r$ input symbols for the inner
encoder. In general, the inner encoder can be oblivious to the fact
that the input symbols might contain already encoded data.

For the initial implementation we required all source symbols to be
available before any encoding could take place. This is however not
necessary in cases where both encoders support systematic encoding. In
such cases, it would be possible to push the initial $n$ symbols
directly through both encoders without doing any coding operations or
adding additional delay.  An illustration of this is shown in
Fig.~\ref{fig:systematic_outer_systematic_inner}, where a set of $n =
8$ original symbols are sent with the outer encoder configured to
build an expansion of $r = 2$.

\begin{figure}
  \centering
\subfloat[ \label{fig:systematic_outer_systematic_inner}]{{\includegraphics[width=0.35\linewidth]{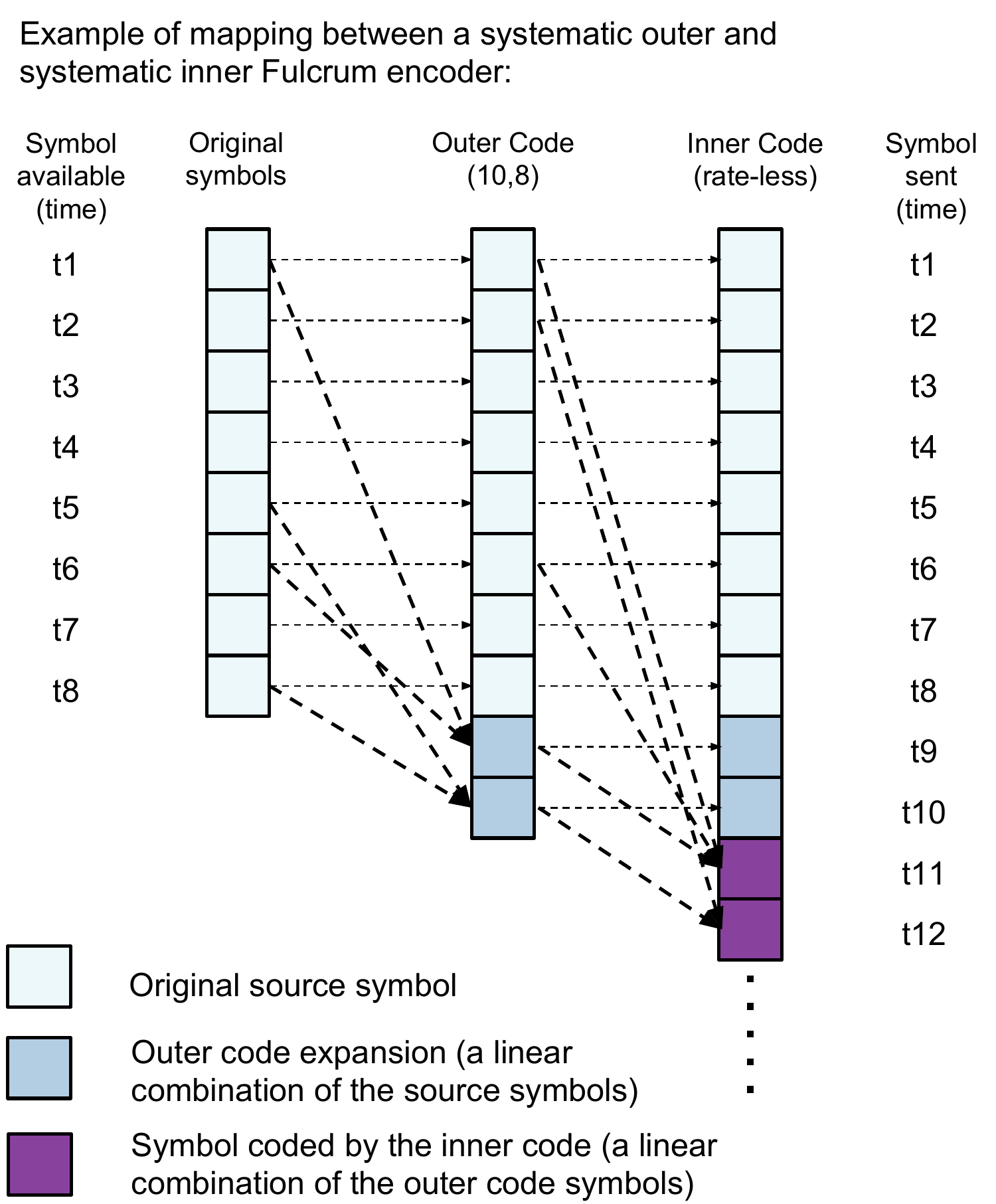}
 }}%
\qquad
    \subfloat[\label{fig:non_systematic_outer_non_systematic_inner} ]{{\includegraphics[width=0.35\linewidth]{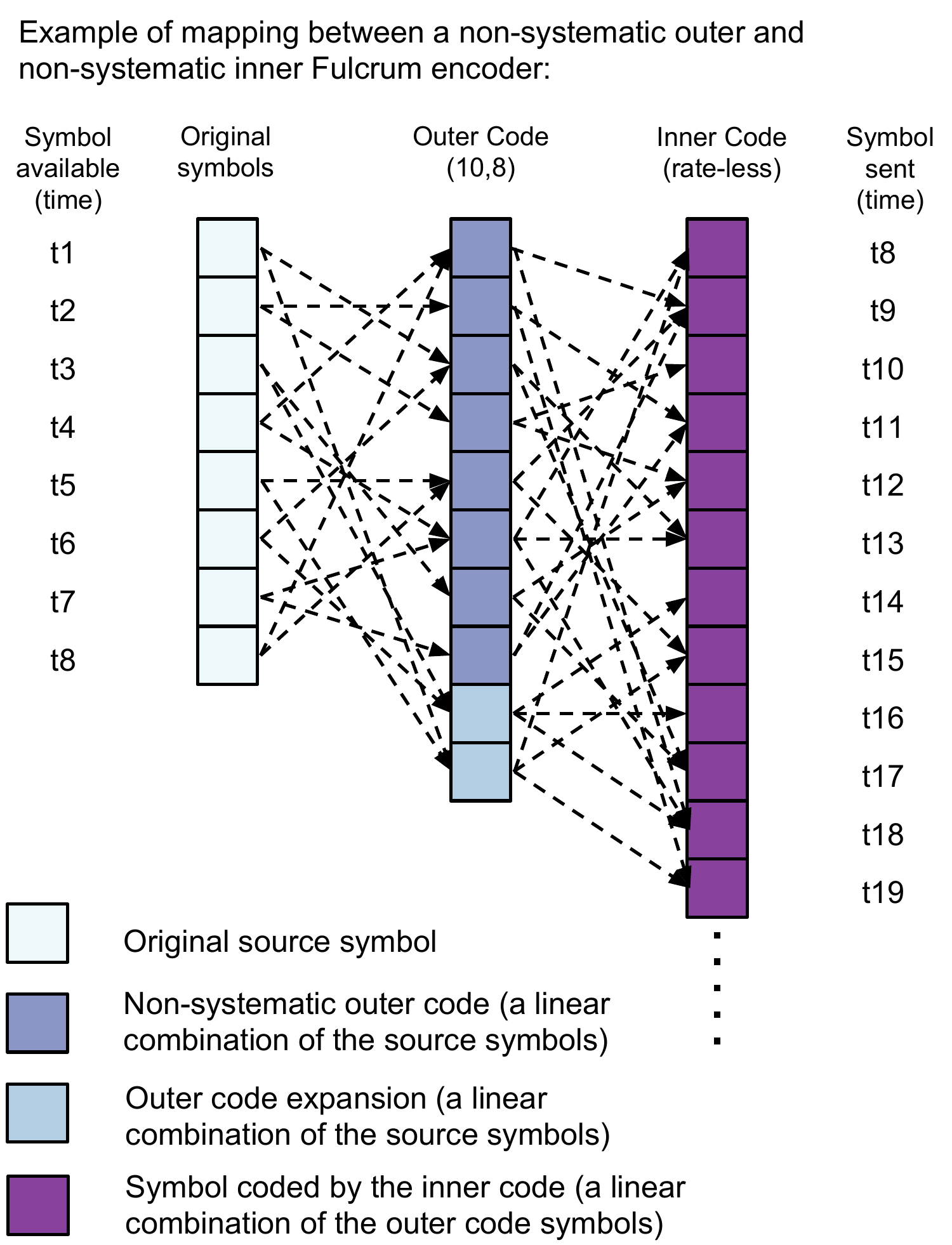}
      }}%
    \caption{Example of (a) a systematic outer encoder and a
      systematic inner encoder with $n=8$ and $r=2$, and (b) a
      non-systematic outer encoder and a non-systematic inner encoder
      with $n=8$ and $r=2$ }
\vspace{-0.3cm}\end{figure}

As both encoders in Fig.~\ref{fig:systematic_outer_systematic_inner}
are systematic, no coding takes place until step $9$ and $10$, where
the outer encoder produces the first encoded symbols. At this point,
the inner encoder is still in the systematic phase and therefore
passes the two symbols directly through to the network. In step $11$,
the inner encoder also exits the systematic phase and starts to
produce encoded symbols. At this stage, the inner encoder is fully
initialized and no additional symbols are needed from the outer
encoder, all following encoding operations therefore take place in the
inner encoder.  

As shown in this simple example, using a systematic structure in both
encoders can be very beneficial for low delay applications because
packets can be sent as they arrive at the encoder. Systematic encoding
is not always required for attaining this low delay.
For example, if the inner encoder is a standard RLNC encoder only
generating non-zero coefficients for the available symbols, i.e.,
using an on-the-fly encoding mechanism.

In the case of a non-systematic inner code, this low delay
performance is typically not possible. However, there are several
applications where non-systematic encoding may be more beneficial,
e.g., for security, multiple-source and/or multiple-hop networks.  For
data confidentiality, using a systematic outer code becomes a weakness
in the system. In this case, a dense, high field outer code is key to
providing the required confidentiality.

As an example,
Fig.~\ref{fig:non_systematic_outer_non_systematic_inner} shows the use
of a non-systematic outer encoder. Assuming the outer mapping is kept
secret, only nodes with knowledge of the secret would be able to
decode the actual content. Whereas all other nodes would still be able
to operate on the inner code.  For multi-hop networks or multi-source
networks, a systematic inner code may not be particularly useful for
the receiver as the systematic structure will not be preserved as the
packet traverses the network and is
recoded. Fig. ~\ref{fig:non_systematic_outer_non_systematic_inner}
shows that it is also possible to use a non-systematic encoding scheme
at the inner encoder. This is typically implemented to minimize the
risk of transmitting linear dependent information in networks which
may contain multiple sources for the same data, e.g. in Peer-to-Peer
systems, or if the state of the sinks is unknown.

\subsection{Implementation of the Decoder}
We have developed decoders supporting
all three types of receivers mentioned in
Section~\ref{sec:description}.

\textbf{Outer decoder:} immediately maps from the inner to the outer
code essentially decoding in $GF(2^h)$. This type of decoder is shown
in Fig.~\ref{fig:outer_decoder_detail}. In order to perform this
mapping, a small lookup table storing the coefficients was used. The
size of the lookup table depends on whether the outer encoder is
systematic or not. In the case of a systematic outer encoder a lookup
table of size $r$ is sufficient, since the initial $n$ symbols are
uncoded (i.e. using the unit vector). However, in case of a
non-systematic outer encoder all $n+r$ outer encoding vectors needs to
be stored. An alternative approach would be to use a pseudo-random
number generator to generate the encoding vectors on the fly as
needed. One advantage of the lookup table is that it may be precomputed
and therefore would not consume any additional computational resources
during encoding/decoding.

\textbf{Inner decoder:} decodes using only $GF(2)$ operations, requiring
  a systematic outer encoder (Fig.~\ref{fig:inner_decoder_detail}). In
  this case, the decoder's implementation is very similar to a
  standard RLNC $GF(2)$ decoder configured to received $n+r$
  symbols. The only difference being that only $n$ of the decoded
  symbols will contain the original encoded data. If sparse inner
  codes are used, other decoding algorithms could be used, e.g.,
  belief propagation~\cite{Shokrollahi06}.

\textbf{Combined decoder:} attempts to decode as much as
  possible using the inner decoder before switching to the typically
  more computationally costly outer decoder. Note that this type of
  decoding only is beneficial if the outer encoder is systematic or,
  potentially, very sparse. Otherwise, the combined decoder gives no
  advantages in general.
\begin{figure}[t]
\centering
\includegraphics[width=0.35\textwidth]{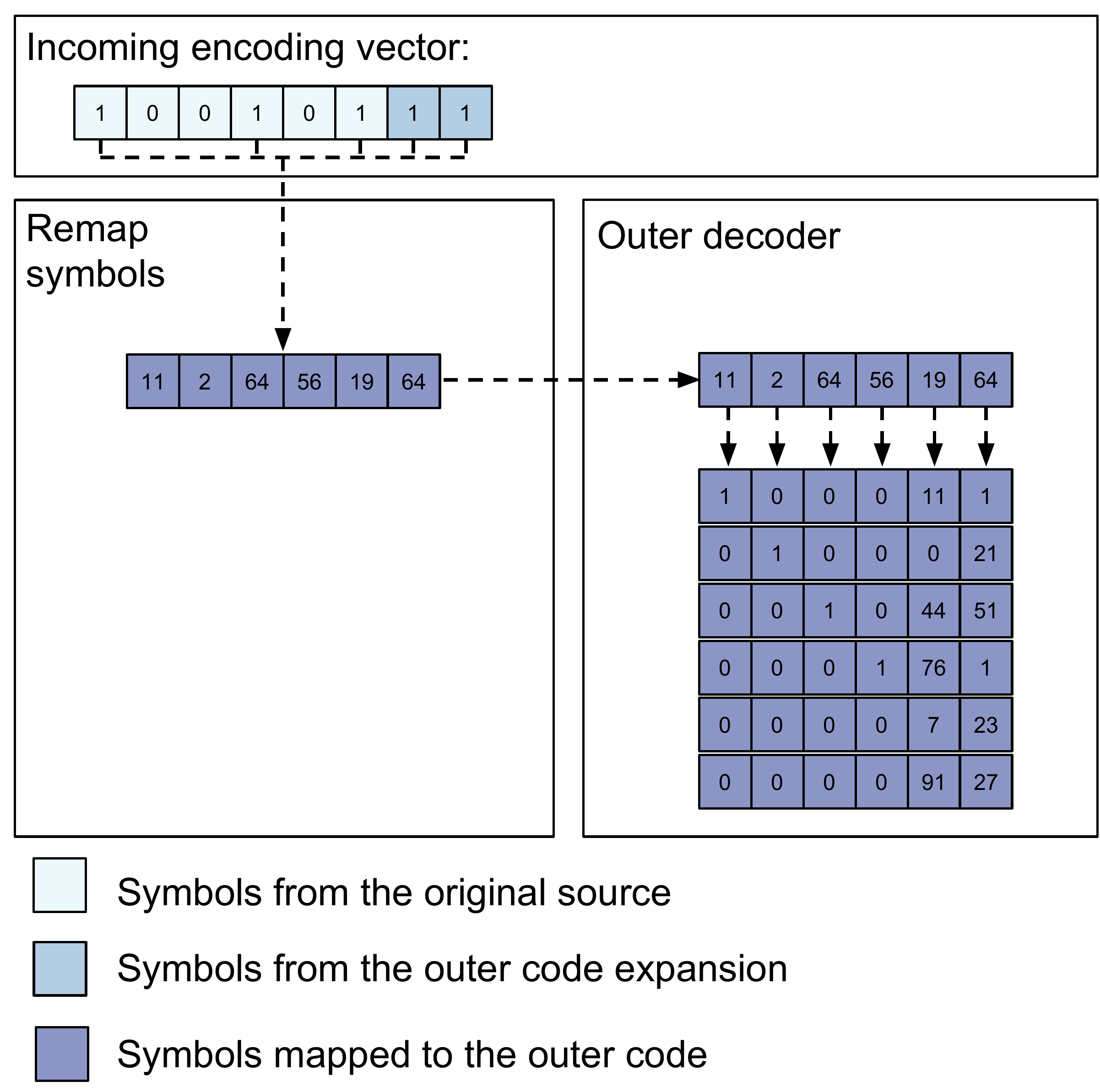}
\caption{Decoding process in a Fulcrum outer decoder. The inner code
  encoding vector is mapped directly back to the outer field. Then, a
  standard decoder can be used to decode the
  data.}
\label{fig:outer_decoder_detail}
\vspace{-0.3cm}
\end{figure}

\begin{figure}
  \centering
\subfloat[ \label{fig:inner_decoder_detail}]{{\includegraphics[width=0.35\linewidth]{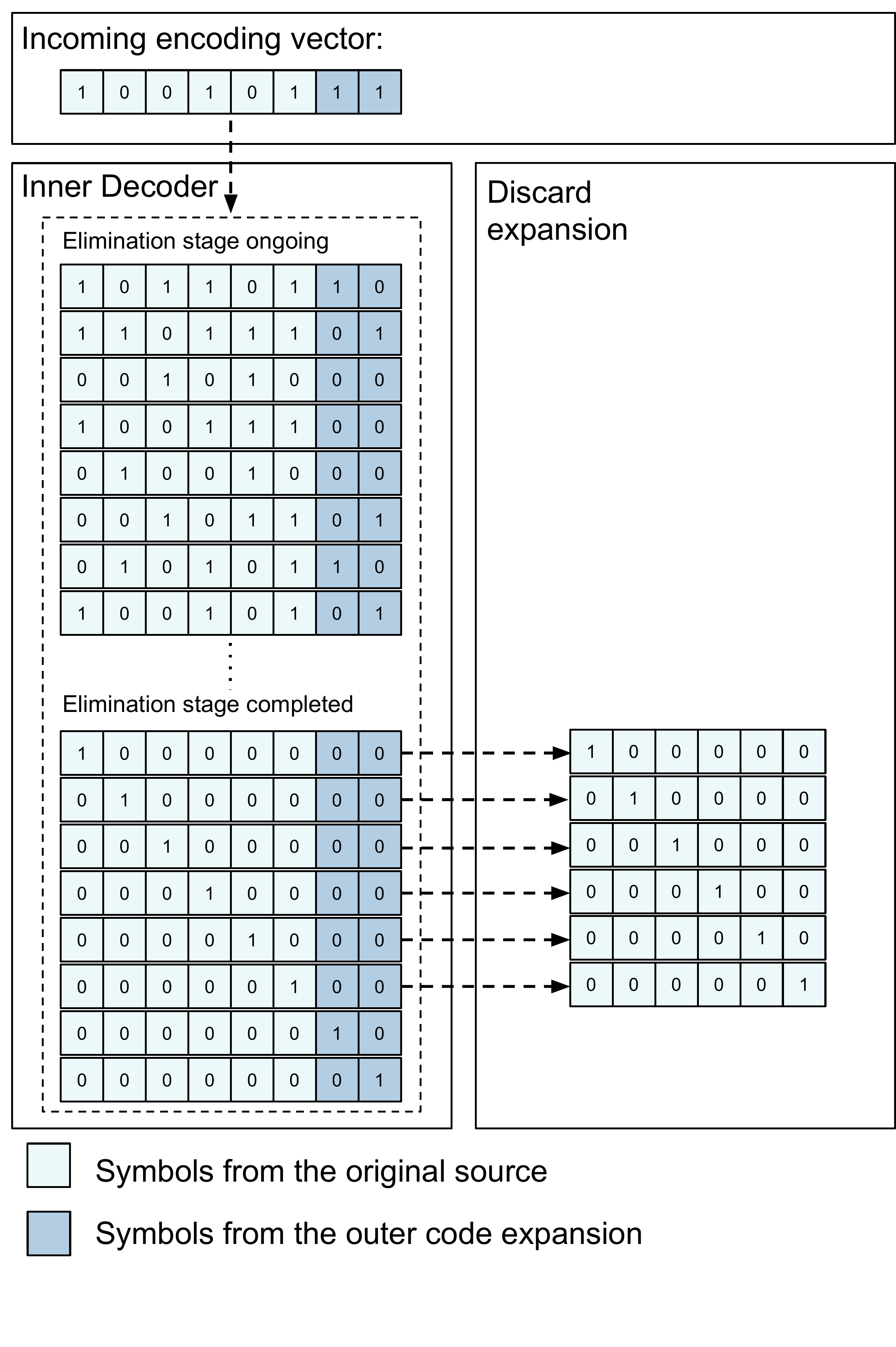}
 }}%
\qquad
    \subfloat[\label{fig:combined_decoder_detail} ]{{\includegraphics[width=0.35\linewidth]{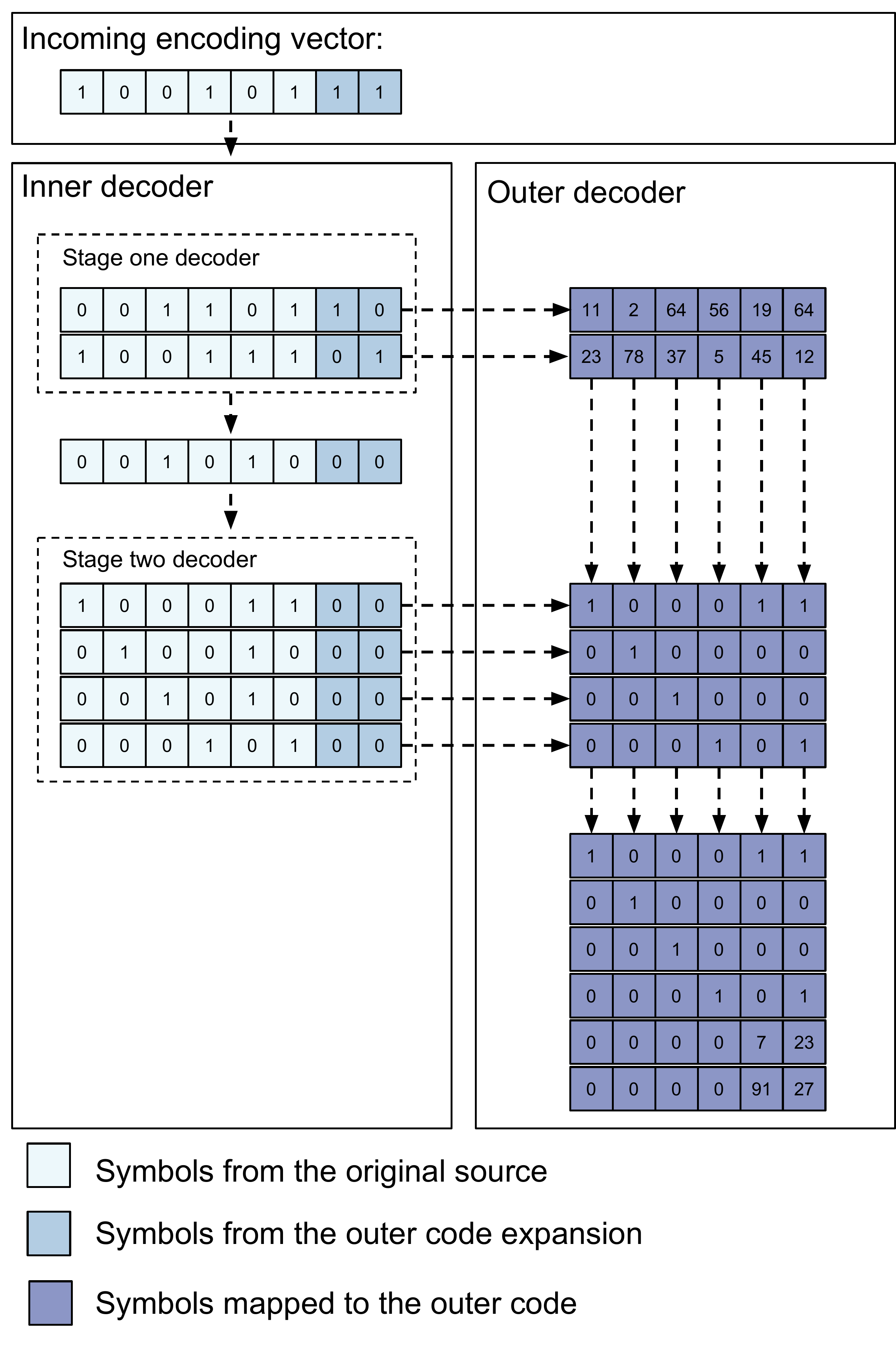}
      }}%
    \caption{Examples of Fulcrum's (a) inner decoder and (b) combined
      decoder. A Fulcrum inner decoder (a) skips the use of the outer
      code by decoding the entire inner block and then discarding all
      symbols belonging to the outer expansion.  A Fulcrum combined
      decoder (b) uses a two stage inner decoder to eliminate as much
      of the contribution of the outer code as possible before mapping
      the symbol to the outer decoder. It should be noted that this
      only works in cases where the outer code is systematic.}
\vspace{-0.3cm}
\end{figure}



In order to understand how this works, let us go through the example
shown in Fig.~\ref{fig:combined_decoder_detail}. When an encoding
vector arrives at a combined decoder, it is first passed to the inner
decoder. Internally, the inner decoder is split into two stages. In
stage one, we attempt to eliminate the extension added in the outer
encoder (these are the symbols that when mapped to the outer decoder
will have coding coefficients from the outer field). If stage one
successfully eliminates the expansion, the symbol is passed to stage
two. In the stage two decoder, we only have linear combinations of
original source symbols. These symbols have a trivial encoding vector
when mapped to the outer decoder. Once stage one and stage two
combined have full rank the stored symbols are mapped to the outer
decoder. Notice in Fig.~\ref{fig:combined_decoder_detail} how symbols
coming from stage two have coding coefficients $0$ or $1$ require only
a few operations to be decoded, whereas the symbols coming from stage
one have a dense structure with coding coefficients coming from the
outer field, represented by $c_{xy} \in GF(2^h)$, where $GF(2^h)$ is
the field used for the outer code. After mapping to the outer decoder,
the final step is to solve the linear system shown in the lower right
of Figure~\ref{fig:combined_decoder_detail}.

\vspace{-0.2cm}
\section{Performance Results}
In the following section, we present performance results obtained by
running the different algorithms on various devices as depicted in
Table~\ref{tbl:devices}. For all benchmarks a packet size of 1600~B
was used and the outer Fulcrum code is performed over $GF(2^8)$. We
implemented the Fulcrum encoder and the three decoder types in
Kodo~\cite{kodo11}. The results for the RLNC encoders and decoders in
$GF(2)$ (``Binary'' in the Figures) and $GF(2^8)$ (``Binary8'' in the
Figures) use the current implementation in this library with and
without Single Instruction Multiple Data (SIMD) operations for
hardware speed up. Performance is evaluated using Kodo's
benchmarks. We assume a systematic outer code structure to compare
performance of the three decoder types.

\begin{table}[ht!]
\centering
\resizebox{0.5\columnwidth}{!}{
\begin{tabular}{lll}
    \hline
    \textbf{Alias} & \textbf{Device} & \textbf{CPU} \\ \hline
    N6     & Nexus 6 & Quad-core 2.7 GHz Krait 450 \\ 
    N9     & Nexus 9 & Dual-core 2.3 GHz Denver \\ 
    i5     & Intel NUC D54250WYK & Dual-core 2.6 GHz Intel core i5-4250U \\ 
    i7     & Dell latitude E6530 & Quad-core 2.7 GHz Intel core i7-3740QM \\  
    Rasp v2 & Raspberry PI 2 model B V1.1 & Quad-core 900MHz ARM Cortex-A7 CPU \\
    S5     & Samsung S5 & Quad-core 2.5 GHz Krait 400 \\ \hline
\end{tabular}
}
\caption{Measured devices}
\label{tbl:devices}
\vspace{-0.2cm}
\end{table}


Figure~\ref{fig:fulcrum_proc_speed_decoder} shows the decoding
throughput for the Fulcrum code, with different $r$ and decoder type.
This is compared against the performance of an \ac{RLNC} decoder using
$GF(2)$, as this represents the fastest dense code, and an \ac{RLNC}
decoder using $GF(2^8)$, as this represents a commonly used dense code
with the same field size used in the outer code of the Fulcrum schemes
and where decoding probability approaches 1 when $n$ packets have been
received.

When only the inner code over $GF(2)$ is utilized for decoding in
Fulcrum (inner decoder), Fulcrum is similar to \ac{RLNC} over
$GF(2)$. When only the outer code over $GF(2^8)$ is utilized in
Fulcrum (outer decoder), Fulcrum becomes similar to \ac{RLNC} over
$GF(2^{8})$.  Thus, in these two cases the decoding throughput for
Fulcrum is expected to be equivalent to \ac{RLNC} over $GF(2)$ and
\ac{RLNC} over $GF(2^{8})$, respectively. This is confirmed and
verifies that the decoding implementation performs as expected in
these two known cases.
The case of the combined decoder is more interesting as it shows the
gain over \ac{RLNC} with $GF(2^{8})$. Not only is the Fulcrum combined
decoder always faster compared to \ac{RLNC} over $GF(2^{8})$, but the
performance also approaches that of \ac{RLNC} over $GF(2)$ as the
generation size grows. For $n=1024$ packets, the combined decoder is
$20$ times faster than the \ac{RLNC} over $GF(2^{8})$, but with
similar decoding probability.

\begin{figure}
  \centering
\subfloat[ \label{fig:fulcrum_proc_speed_decoder}]{{\includegraphics[width=0.45\linewidth]{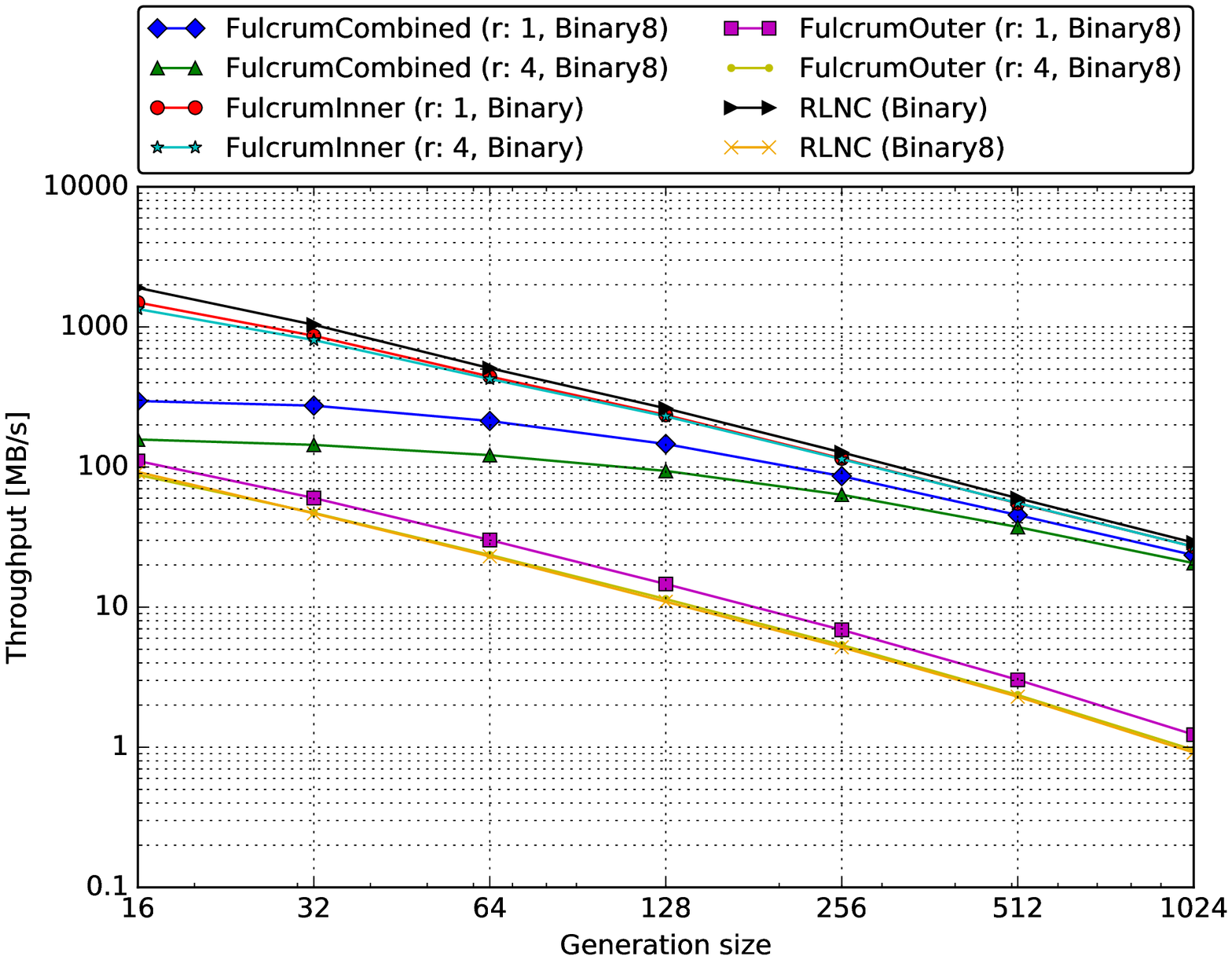}
 }}%
\qquad
    \subfloat[\label{fig:fulcrum_proc_speed_encoder} ]{{\includegraphics[width=0.45\linewidth]{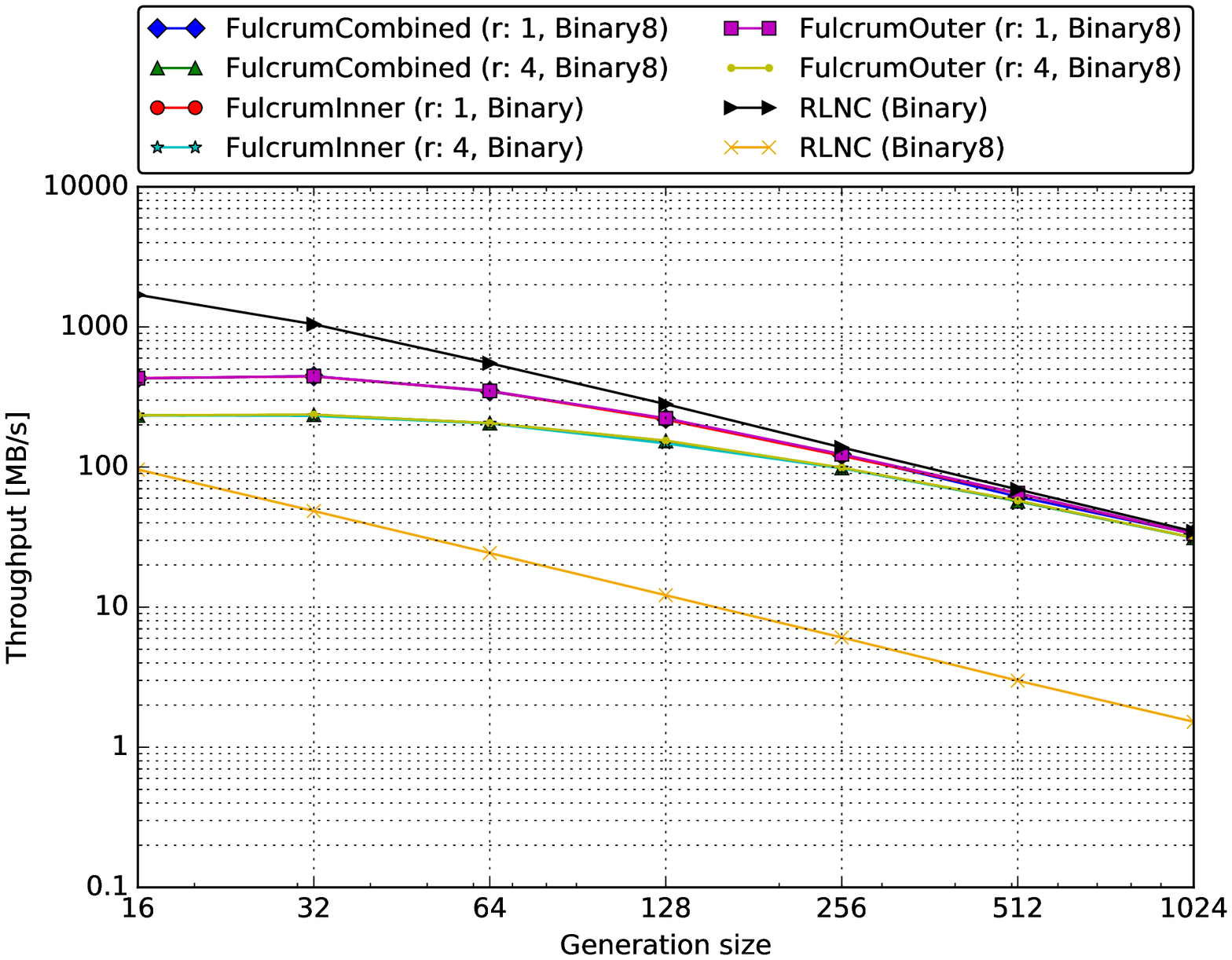}
      }}%
    \caption{Processing speed of an i7 without SIMD optimizations for
      (a) decoding of various Fulcrum decoders compared to RLNC
      decoders, and (b) encoding of Fulcrum compared to RLNC encoding.
      The encoding speed of Fulcrum does not depend on the decoder
      type (combined, inner, or outer). }
\vspace{-0.3cm}
\end{figure}


Figure~\ref{fig:fulcrum_proc_speed_decoder} shows that the combined
decoder has some performance dependence for small generation size for
different $r$ values, namely, the higher the $r$ the lower the
processing speed, but faster than the outer decoder. However, this
difference in performance becomes negligible as the number of data
packets per generation increases. The reason is that most of the
processing effort will be spent decoding in the inner code, and the
effect of the $r$ expansion packets is less marked.  Decoding speed is
usually given a higher priority than the encoding speed, e.g., if
there are more decoders than encoders, or because the decoding process
tends to be slower than the encoding one. However, encoding speed can
be critical in some cases, e.g., a satellite transmitting to an earth
station, sensor nodes collecting and sending data to a base station,
because there is an inherent constraint on the sender's computational
capabilities or energy.  Figure~\ref{fig:fulcrum_proc_speed_encoder}
shows the encoding speed compared to the baseline \ac{RLNC} over
$GF(2)$ and $GF(2^8)$.  For the case of $n = 16$ packets in the
generation, the Fulcrum encoder runs $3.2$x to $6.6$x faster for $r =
1$ and $r=4$, respectively, compared to the $GF(2^8)$ \ac{RLNC}
encoder.  As $n$ increases, so does the gain over the \ac{RLNC}
$GF(2^8)$ and the dependency on the choice of $r$ decreases.  For
example, at $n = 128$ packets the Fulcrum encoder is approximately
$14$x faster than the \ac{RLNC} $GF(2^8)$ encoder, and for $g=256$ the
encoding speed is close to \ac{RLNC} over $GF(2)$.

\begin{figure}
  \centering
\subfloat[ \label{fig:fulcrum_SIMD_i7_decoder}]{{\includegraphics[width=0.45\linewidth]{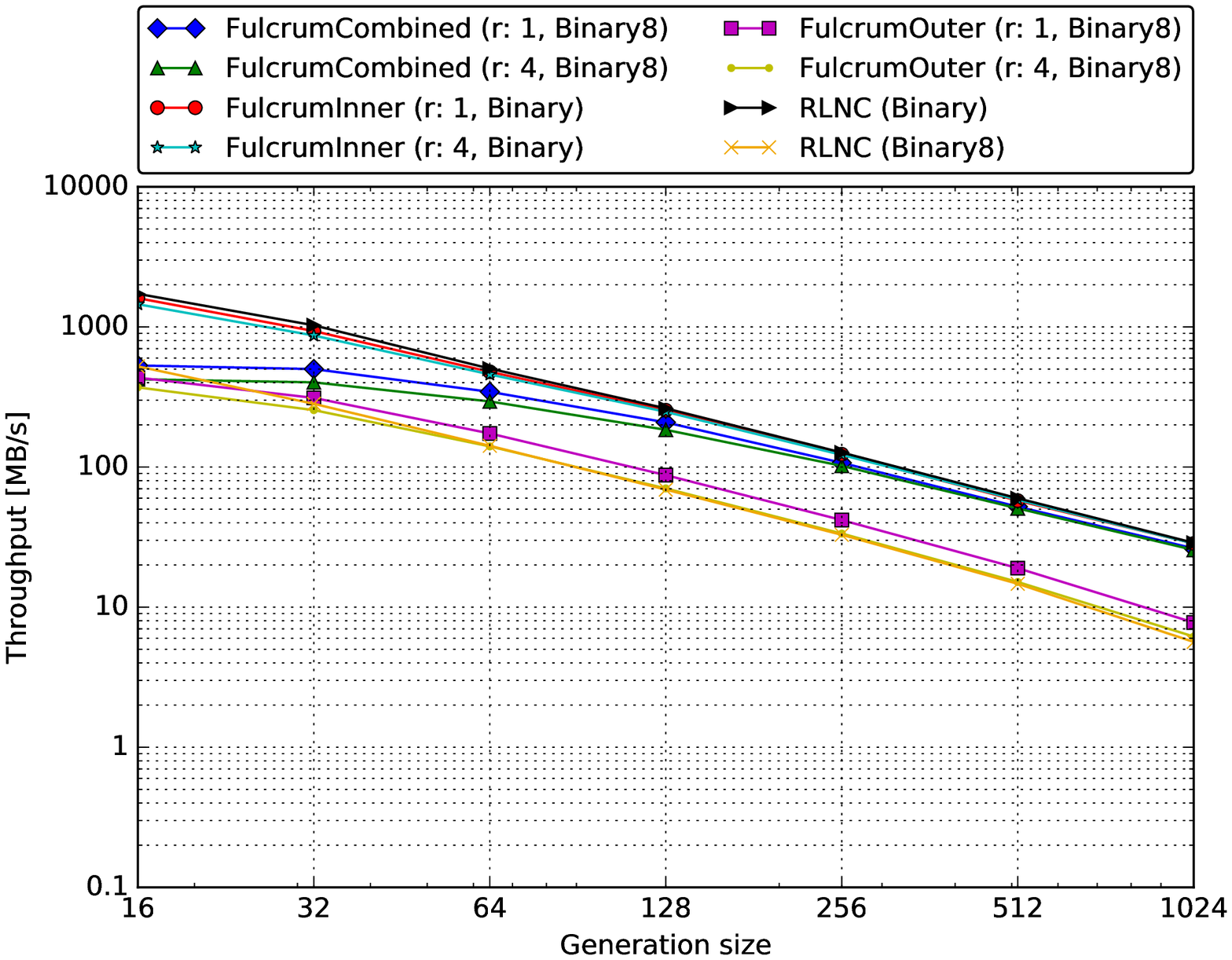}
 }}%
\qquad
    \subfloat[\label{fig:fulcrum_SIMD_N6_decoder} ]{{\includegraphics[width=0.45\linewidth]{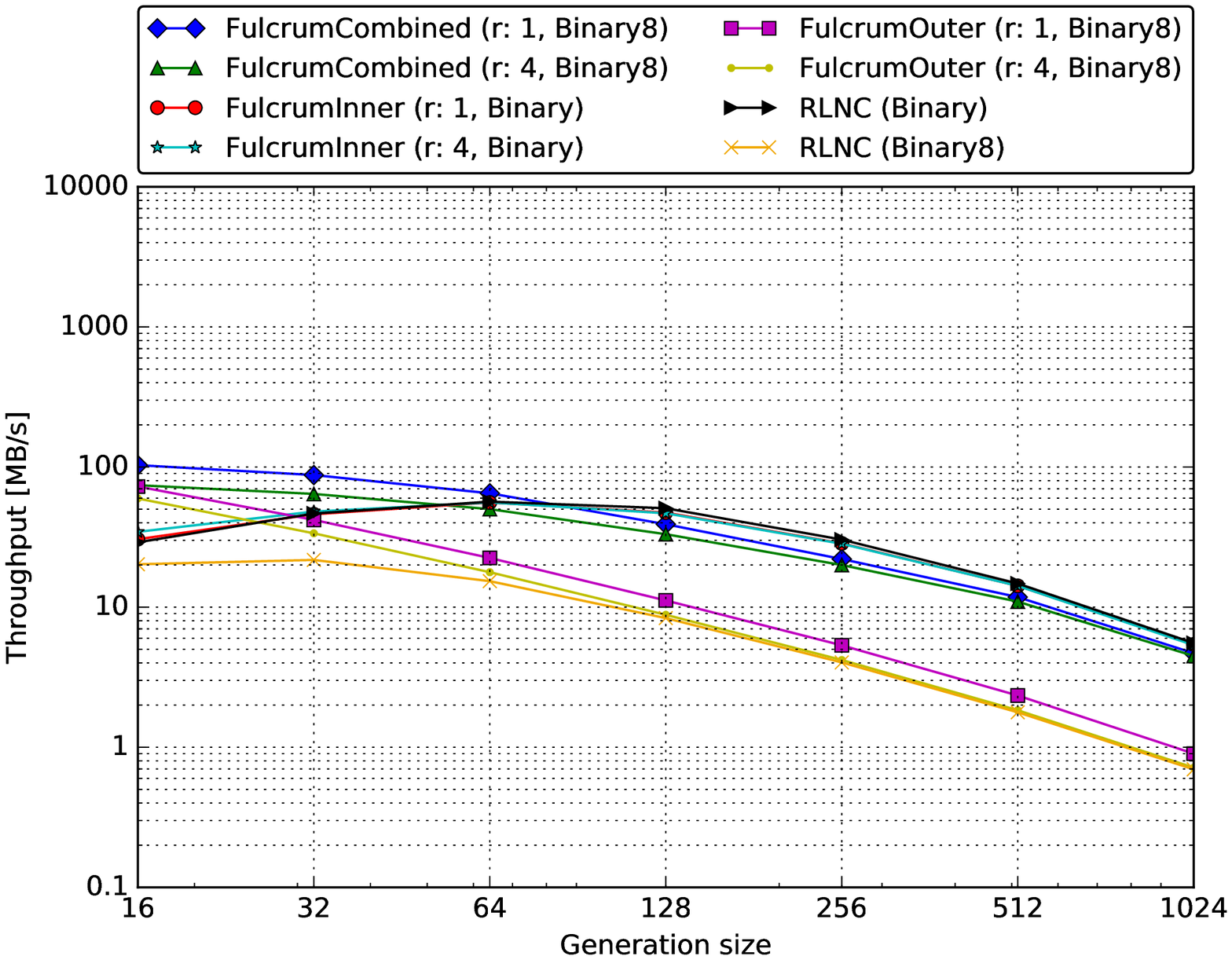}
      }}%
    \caption{Processing speed of decoding with SIMD optimizations for
      (a) the i7 and (b) the N6 }
\vspace{-0.3cm}
\end{figure}

Figure~\ref{fig:fulcrum_SIMD_i7_decoder} studies the effect of SIMD
instructions on the performance of the decoders on the i7. This
optimization has an effect on operations with $GF(2^8)$, which means
that the RLNC $GF(2)$ decoder and the inner decoder will show no
changes from
Figure~\ref{fig:fulcrum_proc_speed_decoder}. Figure~\ref{fig:fulcrum_SIMD_i7_decoder}
shows similar trends, but with reduced gains with respect to RLNC over
$GF(2^8)$ as the performance of $GF(2^8)$ operations are greatly
improved. Nonetheless, at $n = 128$ the decoding speed is $2.8$ times
higher than RLNC $GF(2^8)$ and close to the performance of RLNC over
$GF(2)$.

Figure~\ref{fig:fulcrum_SIMD_N6_decoder} shows the performance of the
decoders with SIMD optimizations on a N6 mobile device. The trends in
the N6 are more surprising, as standard RLNC decoders for $GF(2)$ are
slower than the Combined Decoders for $n \leq 64$ (up to 3.3 times
slower at $n = 16$).  This gain of the combined decoder is similar
when compared to Inner decoders, which use the same algorithm as RLNC
$GF(2)$ for decoding purposes. This means that the combined decoder is
not only faster than RLNC in $GF(2^8)$ in the N6, but also faster than
RLNC in $GF(2)$ for the small and medium range of
$n$. Figure~\ref{fig:fulcrum_SIMD_N6_decoder} also shows that the
Fulcrum outer decoder can also be faster than RLNC $GF(2^8)$ for $n
\leq 64$ packets. Similar trends have been observed for the S5 and the
N9, showing that Fulcrum does not require all devices to trade-off
speed for performance (or vice-versa) but can provide better
performance in both domains.


\begin{figure}
  \centering
\subfloat[ \label{fig:fulcrum_NoSIMD_combined_decoder_alldevices}]{{\includegraphics[width=0.45\linewidth]{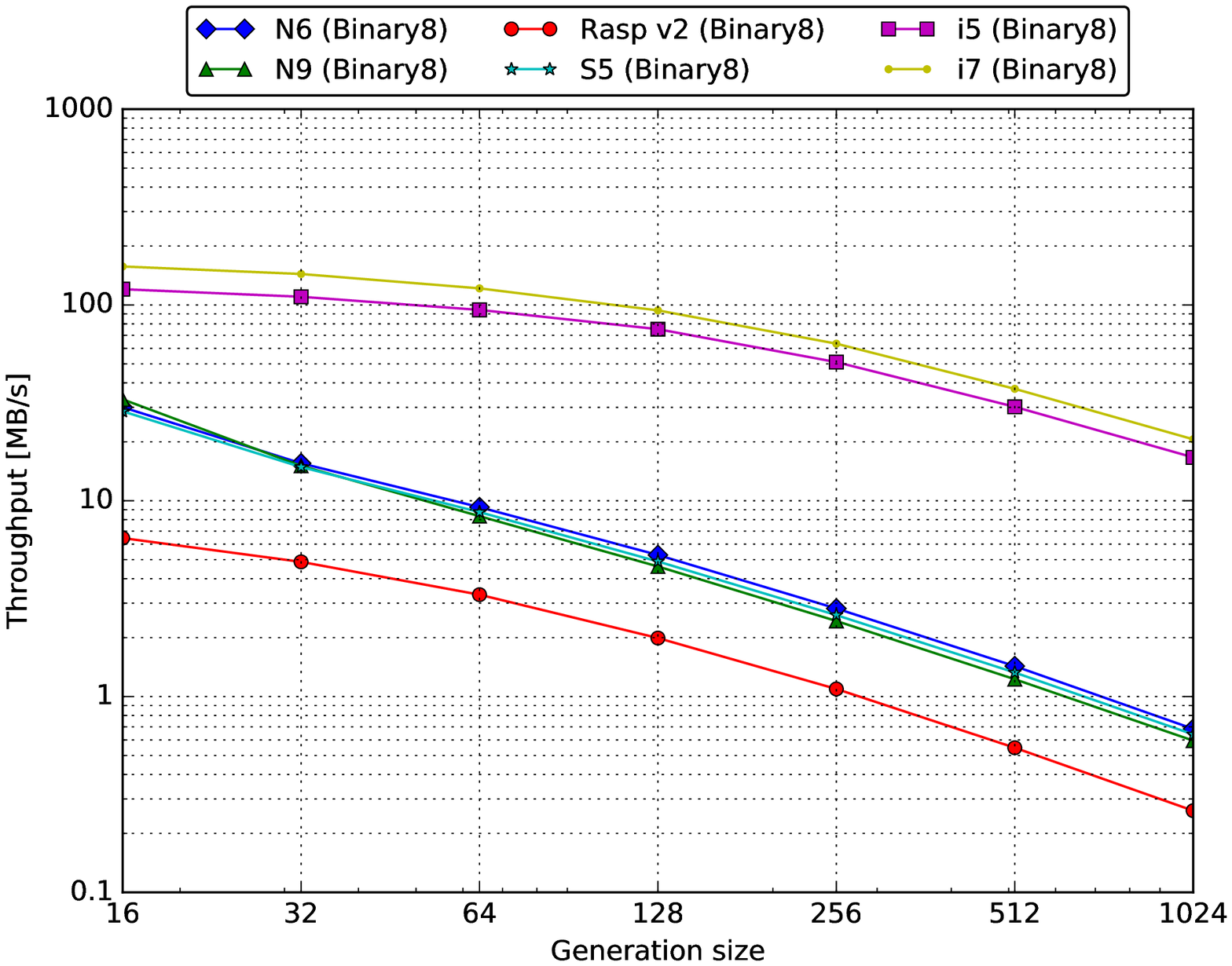}
 }}%
\qquad
    \subfloat[\label{fig:fulcrum_SIMD_combined_decoder_alldevices} ]{{\includegraphics[width=0.45\linewidth]{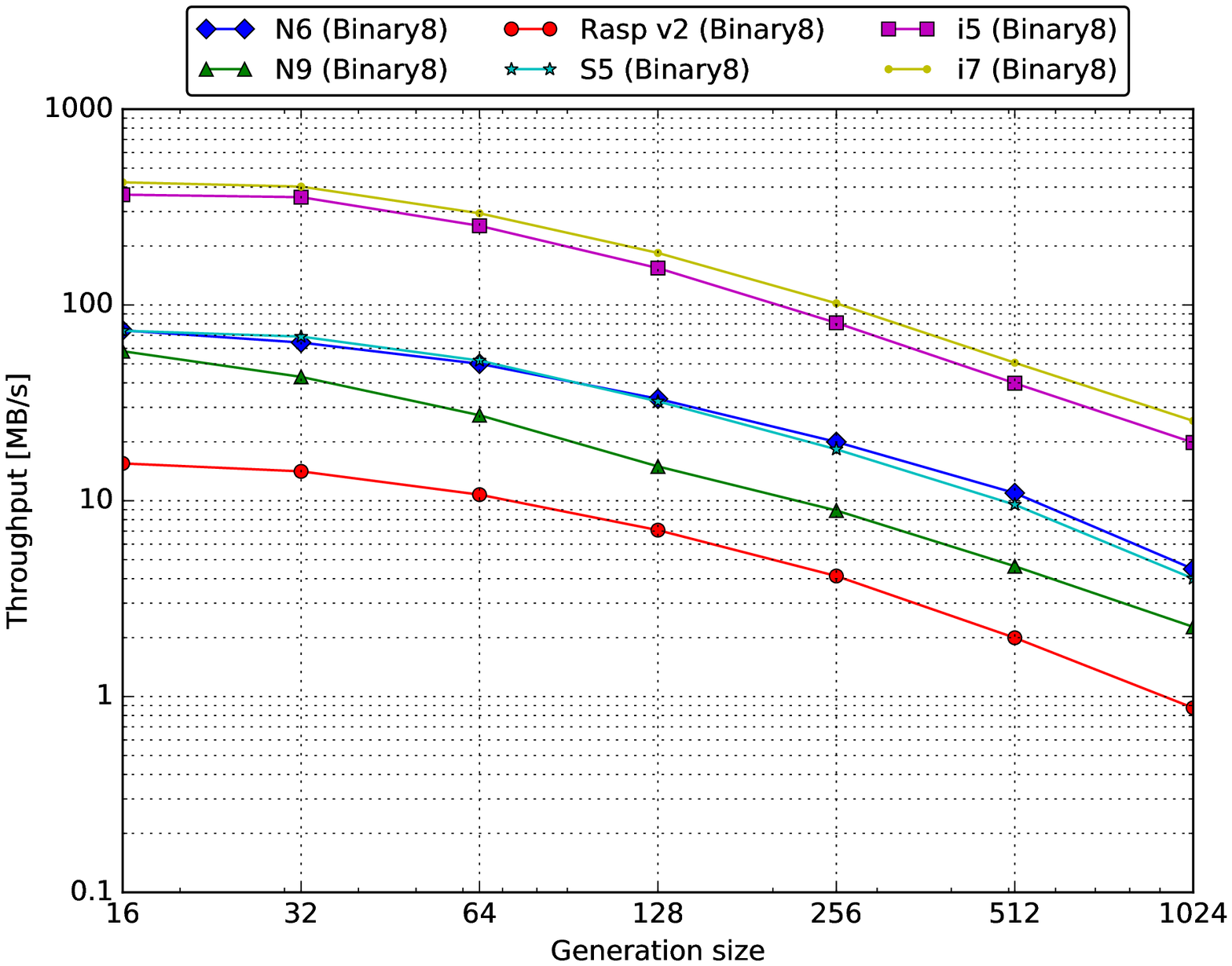}
      }}%
    \caption{Processing speed of Combined Decoder (a) without SIMD and
      (b) with SIMD for different devices }
\vspace{-0.3cm}
\end{figure}

Finally, Figures~\ref{fig:fulcrum_NoSIMD_combined_decoder_alldevices}
and~\ref{fig:fulcrum_SIMD_combined_decoder_alldevices} show the
performance of the combined decoder on various commercial devices
without and with SIMD optimizations, respectively. The use of SIMD
results two- to four-fold speed ups for all devices for low
$n$. However, we observe that the benefit of SIMD is more marked for
mobile devices at high $n$ (up to five times the speed up), while for
desktops the gains over no SIMD have essentially dissapeared. This
does not mean that SIMD in desktops is not improving computation of
$GF(2^8)$, but rather that the main limitation becomes the processing
of $GF(2)$ operations at high speeds. More important is the fact that
mobile devices, which are more energy and computationally limited, can
benefit significantly and over a large range of $n$ from the
combination of Fulcrum and the use of apropriate SIMD instructions.

\vspace{-0.2cm}
\section{Conclusions}
This paper presents Fulcrum network codes, an advanced network code
structure that preserves \ac{RLNC}'s ability to recode seamlessly in
the network while providing key mechanisms for practical deployment of
network coding. Fulcrum addresses several of the standing practical
problems with existing \ac{RLNC} codes and rateless codes, by
employing a concatenated code design. This concatenated code design
provides our solution with a highly flexible, tunable and intuitive
design. This paper describes in detail the design of Fulcrum network
codes and its practical benefits over previous network coding designs
and it provides mathematical analysis on the performance of Fulcrum
network codes under a wide range of conditions and scenarios. The
paper also presents a first implementation of Fulcrum in the Kodo C++
network coding library as well as benchmarking its performance to
high-performance \ac{RLNC} encoder and decoders.

Our throughput benchmarks show that Fulcrum provides much higher
encoding/decoding processing speed compared to \ac{RLNC} $GF(2^8)$. In
fact, the processing speeds approach those of \ac{RLNC} $GF(2)$ as the
generation size grows. More importantly, Fulcrum can maintain the
decoding probability performance of \ac{RLNC} $GF(2^8)$ at the same
time that the processing speed is increased by up to a factor of $20$ in
some scenarios with our initial implementation. Furthermore, the
trade-off between coding processing speed and decoding probability can
easily be adjusted using the outer code expansion $r$ to meet the
requirements of a given application.

Fulcrum solves several standing problems for existing \ac{RLNC}
codes. First, it enables an easily adjustable trade-off between coding
throughput and decoding probability. Second, it provides a higher
coding processing speed when compared to the existing \ac{RLNC} codes
in use. Third, it reduces the overhead associated with the coding
vector representation, necessary for recoding, while maintaining a
high decoding probability. Fourth, it reduces the type of operations
and logic that the network needs to support while allowing end-to-end
devices to tailor their desired service and performance, making a key
step to widely deploying network coding in practice. This has an added
advantage of allowing the network to support future designs seamlessly
and naturally providing backwards compatibility.

Given these advantages, Fulcrum is particularly well suited for a wide
range of scenarios, including (i) distributed storage, where the
reliability requirements are high and many storage units are in use;
(ii) wireless (mesh) networks, where the packet size is typically
small, which results in large generation sizes when large file are
transmitted; (iii) heterogeneous networks, since Fulcrum supports
different decoding options for (computationally) strong and weak
decoders; (iv) wireless sensor networks, where the packet sizes are
small requiring small overhead and also the devices are energy- and
computationally-limited.  Future work will study optimal solutions to
use Fulcrum's structure to spread the complexity over the network.

\vspace{-0.3cm}
\bibliographystyle{IEEEtran}

\bibliography{bibtex}

\end{document}

%% file: acronym.tex
\begin{acronym}[]
\acro{ACK}{Acknowledgment}
\acro{AMR}{Adaptive Multi-Rate}
\acro{AP}{Access Point}
\acro{ARQ}{Automatic Repeat-reQuest}
\acro{ARF}{Auto Rate Fallback}
\acro{API}{Application Programming Interface}
\acro{BER}{Bit Error Rate}
\acro{BS}{Base Station}
\acro{BSS}{Basic Service Set}
\acro{CFP}{Contention Freeze Periode}
\acro{CPU}{Central Processing Unit}
\acro{CSMA/CA}{Carrier Sense Multiple Access with Collision Avoidance}
\acro{CW}{Contention Window}
\acro{CWN}{Cooperative Wireless Network}
\acro{CTS}{Clear to Send}
\acro{DCF}{Distributed Coordination Function}
\acro{DIFS}{Distributed Interframe Space}
\acro{ERP}{Extended Rate PHY}
\acro{FEC}{Forward Error Correction}
\acro{FCS}{Frame Check Sequence}
\acro{FF}{Finite Field}
\acro{GF}{Galois Field}
\acro{GPS}{Global Positioning System}
\acro{HTTP}{Hyper Text Transfer Protocol}
\acro{IBSS}{Independent BSS}
\acro{IC}{Interference Cancellation}
\acro{ICST}{Institute for Computer Sciences, Social-Informatics and Telecommunications Engineering}
\acro{IFS}{Interframe Space}
\acro{IP}{Internet Protocol}
\acro{IPTV}{Internet Protocol TeleVision}
\acro{JCN}{Journal of Communications and Networks}
\acro{KL}{Kullback-Leibler}
\acro{LAN}{Local Area Network}
\acro{LNC}{Linear Network Coding}
\acro{LOS}{Line Of Sight}
\acro{MAC}{Medium Access Control}
\acro{MANET}{Mobile Ad hoc NETwork}
\acro{MIMO}{Multiple Input Multiple Output}
\acro{MIT}{Massachusetts Institute of Technology}
\acro{MTBF}{Mean Time Between Failure}
\acro{NAV}{Network Allocation Vector}
\acro{NEP}{Nokia Energy Profiler}
\acro{NC}{Network Coding}
\acro{NLOS}{Non Line Of Sight}
\acro{NIC}{Network Interface Card}
\acro{OSI}{Open Systems Interconnect}
\acro{PC}{Personal Computer}
\acro{PDA}{Personal Digital Assistant}
\acro{PEP}{Packet Error Probability}
\acro{PHY}{Physical Layer}
\acro{PLCP}{Physical Layer Convergence Procedure}
\acro{PNC}{Physical layer Network Coding}
\acro{PPDU}{PLCP Protocol Data Unit}
\acro{QoS}{Quality of Service}
\acro{RLNC}{Random Linear Network Coding}
\acro{RTS}{Request to Send}
\acro{SAIC}{Single Antenna Interference Cancellation}
\acro{SIMD}{Single Instruction Multiple Data}
\acro{SNMP}{Simple Network Managment Protocol}
\acro{SNR}{Signal-to-Noise Ratio}
\acro{SIFS}{Short Interframe Space}
\acro{SRLNC}{Sparse Random Linear Network Coding}
\acro{SSE}{Streaming SIMD Extensions}
\acro{SSID}{Service Set Identifier}
\acro{UDP}{User Datagram Protocol}
\acro{UI}{User Interface}
\acro{VANET}{Vehicular Ad-hoc Network}
\acro{VoIP}{Voice over Internet Protocol}
\acro{WBN}{Wireless Broadcast Network}
\acro{WLAN}{Wireless Local Area Network}
\acro{WMN}{Wireless Mesh Network}
\acro{pmf}{Probability Mass Function}
\acro{TBTT}{Target Beacon Transmision Times}
\acro{TCP}{Transmission Control Protocol}
\end{acronym}